\documentclass{article}[11pt]

\usepackage{graphicx}

\usepackage[utf8]{inputenc} 
\usepackage[T1]{fontenc}

\usepackage{amsmath,amsfonts,amssymb,latexsym,epsfig}
\usepackage{mathrsfs}
\usepackage{verbatim}
\usepackage{latexsym}
\usepackage{amsthm}
\usepackage{amssymb}
\usepackage{graphics}
\usepackage[font=small]{caption}
\usepackage{subcaption}
\usepackage{amsbsy}
\usepackage{fullpage}
\usepackage{enumerate}
\usepackage{times}
\usepackage{multirow}
\usepackage{soul}
\usepackage[normalem]{ulem}

\newtheorem{theorem}{Theorem}[section]
\newtheorem{lemma}[theorem]{Lemma}

\newtheorem{proposition}[theorem]{Proposition}

\newtheorem{remark}[theorem]{Remark}
\newtheorem{hypothesis}[theorem]{Hypothesis}

\newtheorem{definition}[theorem]{Definition}

\numberwithin{equation}{section}
\usepackage{amscd}
\usepackage{tikz}
\usepackage{xcolor}
\usepackage[english]{babel}
\usepackage{times}
\usepackage{graphicx}
\usepackage{amsmath,amssymb}
\usepackage{slidesec}
\usepackage{verbatim}
\usepackage{url}
\usepackage{epsf}
\usepackage{amsmath}
\usepackage{graphics}
\usepackage{amsfonts}
\usepackage{amsbsy}
\usepackage{lscape}
\usepackage{enumerate}
\usepackage{amsthm}
\usepackage{amssymb}
\usepackage{exscale}
\usepackage{algorithm}
\usepackage{algorithmic}

\usepackage{amsopn}
\usepackage{bbm}

\newcommand{\R}{\mathbb{R}}
\newcommand{\Rint}{\mathbb{R}_{++}^n}
\newcommand{\cP}{\mathcal{P}}
\newcommand{\cL}{\mathcal{L}}
\newcommand{\Ulb}{U^{\lambda,b}}
\newcommand{\argmin}{\mathrm{arg min}}
\newcommand{\argmax}{\mathrm{arg max}}
\newcommand{\prox}{\mathrm{prox}}

\newcommand\blfootnote[1]{%
  \begingroup
  \renewcommand\thefootnote{}\footnote{#1}%
  \addtocounter{footnote}{-1}%
  \endgroup
}

\newenvironment{@abssec}[1]{%
     \if@twocolumn
       \section*{#1}%
     \else
       \vspace{.05in}\footnotesize
       \parindent .2in
        \ignorespaces 
     \fi}

\newenvironment{keywords}{\begin{@abssec}{\keywordsname}}{\end{@abssec}}
\newcommand\keywordsname{Key words}

\newenvironment{AMS}{\begin{@abssec}{\AMSname}}{\end{@abssec}}
\newcommand\AMSname{AMS subject classifications}

\begin{document}
\title{Efficient Bayesian computation for low-photon imaging problems
}
\author{Savvas Melidonis$^{1}$ \and Paul Dobson$^{2}$ \and Yoann Altmann$^{3}$ \and Marcelo Pereyra$^{1,4}$ \and Konstantinos C. Zygalakis$^{2,4}$}

\maketitle

\begin{abstract}
This paper studies a new and highly efficient Markov chain Monte Carlo (MCMC) methodology to perform Bayesian inference in low-photon imaging problems, with particular attention to situations involving observation noise processes that deviate significantly from Gaussian noise, such as binomial, geometric and low-intensity Poisson noise. These problems are challenging for many reasons. From an inferential viewpoint, low-photon numbers lead to severe identifiability issues, poor stability and high uncertainty about the solution. Moreover, low-photon models often exhibit poor regularity properties that make efficient Bayesian computation difficult; e.g., hard non-negativity constraints, non-smooth priors, and log-likelihood terms with exploding gradients. More precisely, the lack of suitable regularity properties hinders the use of state-of-the-art Monte Carlo methods based on numerical approximations of the Langevin stochastic differential equation (SDE), as both the SDE and its numerical approximations behave poorly. We address this difficulty by proposing an MCMC methodology based on a reflected and regularised Langevin SDE, which is shown to be well-posed and exponentially ergodic under mild and easily verifiable conditions. This then allows us to derive four reflected proximal Langevin MCMC algorithms to perform Bayesian computation in low-photon imaging problems. The proposed approach is demonstrated with a range of experiments related to image deblurring, denoising, and inpainting under binomial, geometric and Poisson noise.\end{abstract}

\begin{keywords}
\newline \hspace*{0.75in}
uncertainty quantification, proximal algorithms.
\end{keywords}


\blfootnote{\textbf{Funding:} This work was supported by the UK Research and Innovation (UKRI) Engineering and Physical Sciences Research Council (EPSRC) grants EP/V006134/1 , EP/V006177/1 and EP/T007346/1, the UK Royal Academy of Engineering under the Research Fellowship Scheme (RF201617/16/31) and the Leverhulme Trust (RF/ 2020-310).}
\footnotetext[1]{School of Mathematical and Computer Sciences, Heriot-Watt University, Edinburgh, EH14 4AS, Scotland, UK (sm2041@hw.ac.uk, m.pereyra@hw.ac.uk).}
\footnotetext[2]{School of Mathematics, University of Edinburgh, Edinburgh, EH9 3FD, Scotland, UK (pdobson@ed.ac.uk, kzygalak@exseed.ed.ac.uk).}
\footnotetext[3]{School of Engineering and Physical Sciences, Heriot-Watt University, Edinburgh, EH14 4AS, Scotland, UK (y.altmann@hw.ac.uk).}
\footnotetext[4]{Maxwell Institute for Mathematical Sciences, Bayes Centre, 47 Potterrow, Edinburgh, Scotland, UK.}

\section{Introduction}

Photon-limited imaging problems arise in modalities that measure the number of photons emitted or reflected by an object or scene of interest. Canonical examples include emission tomographic imaging \cite{Hohage}, fluorescence microscopy \cite{bertero2018inverse,Hohage}, astronomical imaging \cite{bertero2018inverse,starck2007astronomical}, and single-photon light detection and ranging (LIDAR) \cite{Altmann_2016,sparse_single_photon_lidar,7932527,7150537}. While these modalities have traditionally operated in moderately mild regimes, modern photon-limited imaging applications increasingly operate in low-photon to photon-starved regimes, particularly in the context of quantum-enhanced imaging technologies that exploit the particle nature of light and information in individual photons in order to go beyond the limitations of classical imaging paradigms \cite{single}.

An important characteristic of low-photon imaging problems is that their (discrete-valued) data exhibit statistical properties that deviate significantly from the Gaussian statistics encountered in many other imaging modalities. Mildly low-photon problems often exhibit Poisson or compound-Poisson statistics, whereas more challenging photon-starved problems exhibit approximately Bernoulli/binomial or geometric data \cite{7927467,bin_geo}. Poisson-type statistics arise from using sensors capable of discriminating photon-detection events within a given time period, assuming detector dead times can be neglected 
\cite{bin_geo,7932527,7150537}. The more challenging Bernoulli, binomial and geometric statistics arise from sensors that cannot accurately quantify photons beyond the first detection, such as Single-Photon Detectors \cite{single} and in particular Single-Photon Avalanche Diodes (SPADs) \cite{bin_geo,spad,first_photon_imaging,7932527}. 

Accurate recovery of images from low-photon data is difficult for many reasons. First, from an inferential perspective, the data have very limited information about the solution. This leads to image estimation problems that are severely ill-conditioned or ill-posed, and which exhibit high levels of intrinsic uncertainty. Second, from a computational perspective, the data fidelity terms associated with Poisson, binomial and geometric statistics have poor regularity properties (e.g., exploding gradients) and involve non-negativity constraints on the solution space. This makes it difficult to use gradient-based computation algorithms that scale efficiently to large problems \cite{durmus2016efficient,green2015bayesian}. These difficulties can be mitigated to some extent by using tools from proximal optimisation as proposed in \cite{5492199}. Some of these ideas will be revisited in this paper in the context of Bayesian computation, and combined with new strategies for addressing these difficult problems.

The literature considers three main frameworks to solve low-photon imaging problems: Bayesian statistics \cite{7927467,bin_geo,Marnissi2017,Marnissi2016,8683031,zhou2019bayesian}, machine learning \cite{Bregman_priors,pnp_ADMM_restore,proximal_BM3D,ryu2019plugandplay,Photon_Limited_Unrolling}, and the variational framework \cite{convex_approach,5492199,SPIRAL_paper}. This paper focuses on the Bayesian statistical framework because it enables the use of Bayesian decision theory to analyse these difficult imaging problems \cite{robert2007bayesian}. This theory allows deriving estimators, quantifying uncertainty, automatically adjusting unknown model parameters, and performing model selection in the absence of ground truth, all of which are relevant in problems with high intrinsic uncertainty \cite{2018,durmus2016efficient,pereyra2016maximumaposteriori,vidal2020maximum}.

Bayesian computation for low-photon imaging problems is usually addressed in the following three ways: proximal optimisation, variational inference (VI), e.g., variational Bayes (VB) \cite{survey} or Expectation-Propagation (EP) \cite{Ko15,Yao_SIAM2022,Zhang_2019}, and Markov chain Monte Carlo (MCMC) simulation \cite{10.5555/1051451}. Proximal optimisation strategies are mainly used for maximum-a-posteriori (MAP) estimation in Bayesian models that are log-concave, where MAP estimation is a convex problem that can be solved efficiently by using provably convergent and scalable algorithms \cite{convex_approach,5492199,SPIRAL_paper}. Some basic forms of uncertainty visualisation and quantification can also be formulated as convex optimisation problems \cite{pereyra2016maximumaposteriori,Scalable_quanti}. Despite their theoretical rigour, these methods can only support very limited inferences and log-concave priors. Some recent works consider generalisations to data-driven plug-and-play (PnP) priors that can deliver more accurate results, usually at the expense of weaker theoretical guarantees \cite{Bregman_priors,pnp_ADMM_restore,proximal_BM3D,ryu2019plugandplay,Photon_Limited_Unrolling}. Moreover, to support more advanced inferences without significantly increasing computing effort, VI low-photon imaging methods approximate the Bayesian posterior distribution by a tractable surrogate model (e.g., a Gaussian model) useful for approximate inference. VI are often computationally efficient and deliver accurate point estimators, but they are by nature highly problem-specific \cite{Marnissi2017,Marnissi2016} and may be unreliable for uncertainty quantification because of local convergence issues and approximation errors \cite{survey}. EP methods can suffer from similar limitations, and their computation cost is higher, but they can provide more accurate posterior estimates. Conversely, modern MCMC methods provide a general Bayesian computation strategy to perform inference in imaging problems that can support a wide range of inferences and models, with detailed convergence guarantees, albeit at computational cost that is potentially significantly higher than VI and optimisation strategies. To the best of our knowledge, very few works consider MCMC methods for low-photon imaging problems. For Poisson imaging problems, the state of the art MCMC methods are the Poisson hierarchical Gibbs sampler (PHGS) given in  \cite{per_christian_paper}, the primal-dual preconditioned Crank-Nicolson Langevin MCMC algorithm \cite{zhou2019bayesian}, and the split-and-augmented (SPA) Gibbs sampler \cite{8683031}. For Bernoulli, binomial and geometric data, the state of the art are the Gibbs samplers described in \cite{7927467,bin_geo}. These MCMC methods are either highly model-specific, they do not account for the non-negativity constraint in practice or they do not benefit from important recent developments such as acceleration \cite{vargas2020accelerating}.

The aim of this paper is to propose a general computationally efficient MCMC methodology to perform Bayesian inference in low-photon imaging problems, with special attention to models that are log-concave. Similarly to modern MCMC methods for Gaussian imaging problems, the proposed strategy is based on approximations of the Langevin stochastic differential equation (SDE) \cite{cheng2018underdamped,durmus2016efficient,vargas2020accelerating}, which we modify in key ways to enable computation for low-photon imaging problems with non-negativity constraints. More precisely, in this paper we develop a new approach based on reflected Langevin SDEs, for which we develop convergence theory and a series of MCMC algorithms.

The remainder of this paper is organised as follows: Section \ref{sec:main} defines notation and introduces the class of problems and Bayesian models considered. Section \ref{section_3} recalls the Langevin SDE, underlines the several issues preventing their direct application to Bayesian low-photon imaging problems, and presents the proposed Reflected Langevin SDE as well as its theoretical properties. Following on from this, Section \ref{section4} then presents the proposed MCMC samplers and discusses connections with state of the art MCMC algorithms for models with Gaussian likelihoods \cite{durmus2016efficient,goldman2021gradientbased,pereyra2015proximal,vargas2020accelerating}. Section \ref{section_5} demonstrates the proposed approach on a range of challenging experiments related to image deblurring with Poisson data, denoising with binomial data, and inpainting with geometric data, including comparisons with MAP estimation by convex optimisation \cite{5492199} and the SPA sampler \cite{8683031}.


\section{Problem Statement}
\label{sec:main}

\subsection{Observation Models}\label{sec:observationmodels}

We consider imaging inverse problems where we aim to estimate an unknown image $x\in \mathbb{R}^{n}$ from an observation $y\in\mathbb{R}^{m}$ related to $x$ through a statistical model with likelihood function $p(y|x)$. In the problems we are interested in, the recovery of $x$ from $y$ is ill-posed or severely ill-conditioned\footnote{The estimation problem is said to be ill-posed or ill-conditioned when it does not admit a unique solution or it admits a solution that is not stable w.r.t. small perturbations in $y$.} and involves non-negativity constraints on $x$. The typical scenario is image recovery from Poisson distributed observations in imaging applications when the mean photon arrival rate is low. In this case, we acquire a set of discrete photon measurements $y=[y_{1},\ldots,y_{m}]\in \mathbb{N}_{0}^{m}=\{0,1,\ldots\}^{m}$ related to the ground truth image $x$ through the statistical model
\begin{equation}\label{pois_model}
    y|x \sim \mathcal{P}(Ax)\,\,,
\end{equation}
\noindent where $A\in\mathbb{R}_{+}^{m\times n}$, $\R_+^{m\times n}=\{x\in\R^{m\times n}: \min_{i}x_i\geq0\}$, is a linear operator which models the physical properties of the observation process. A common difficulty in applications is that $AA^{T}$ is often poorly conditioned or rank deficient. The respective negative log-likelihood is given by 
\begin{equation}\label{eq:poissonloglikelihood}
f_{\mathcal{P}}(x) = \sum_{i=1}^{m}\left[(A x)_{i}-y_{i}\log ((A x)_{i})+\log(y_{i}!) + \iota_{\mathbb{R}_{++}^{n}}(x)\right]\,\,,
\end{equation}
where $\iota_{\mathbb{R}_{++}^{n}}(\cdot)$ is the indicator function on $\Rint=\{x\in\R^n: \min_{i}x_i>0\}$ that requires $x$ to be positive. Because of the presence of $A$ and the logarithm, the Poisson log-likelihood is not quadratic and often non-separable. Also, $x \mapsto \nabla f_{\mathcal{P}}(x)$ is not necessarily globally Lipschitz continuous. As a result, the estimation of $x$ from $y$ is highly challenging.


In some single-photon imaging applications, the Poisson likelihood assumption does not hold and the relation between the unknown image and the observations is better modeled by binomial or geometric models \cite{bin_geo,first_photon_imaging,7932527,7150537}. For instance, if the detector dead time\footnote{A detector's dead time period corresponds to the period following a photon detection during which additional photons cannot be detected.} can't be neglected but can be reset a the end of a given period (referred to as repetition period here), then the measurements $y=[y_{1},\ldots,y_{m}]\in \mathbb{N}_{0}^{m}$ an be related to the ground image $x\in \mathbb{R}^{n}$ by the statistical model \cite{bin_geo,7932527,7150537}
\begin{equation}\label{bin_model}
    y|x,t \sim \mathcal{B}in\left(t,1-e^{Ax}\right)\,\,,
\end{equation}
\noindent where $\mathcal{B}in(\cdot,\cdot)$ stands for the product of independent binomial distributions, $t=[t_{1},\ldots,t_{n}]\in \mathbb{N}^{n}$ gathers the numbers of repetition periods for each detector, and $A$ is again a linear operator. Then negative log-likelihood is thus given by
\begin{equation}\label{eq:bin_likelihood}
f_{\mathcal{B}}(x) = \sum_{i=1}^{m}\left[-y_{i}\log (1-e^{-(Ax)_{i}})+(t_{i}-y_{i})(Ax)_{i}+ \iota_{\mathbb{R}_{++}^{n}}(x)\right]\,\,.
\end{equation}
In the first-photon imaging context \cite{first_photon_imaging} where the recordings are interrupted for each pixel after the first detection, the number of repetition periods required to record the first detection event in each pixel follows a geometric distribution specified (in vector form) by
\begin{equation}\label{geo_model}
    t|x,y=1 \sim \mathcal{G}eo\left(1-e^{-Ax}\right)\,\,,
\end{equation}
and the negative log-likelihood is given by
\begin{equation}\label{eq:geo_likelihodd}
f_{\mathcal{G}}(x) = \sum_{i=1}^{m}\left[
(t_{i}-1)(Ax)_{i}-
\log \left(1-e^{-(Ax)_{i}}\right)
+ \iota_{\mathbb{R}_{++}^{n}}(x)\right]\,\,\,\,.
\end{equation}
Note that in \eqref{geo_model}, the variables treated as observations are in $t$, while in \eqref{bin_model}, the observations are in $y$ and $t$ is only a fixed experimental parameter. In extreme photon-starved regimes, the binomial and geometric random noise processes are highly non-linear, have very poor signal-to-noise properties and involve log-likelihoods that are not Lipschitz continuous over the non-negative domain. These models differ from the more classical Poisson model but share similar challenges, which is why we consider them as exemplar observation models in this work (to illustrate that our general samplers are not likelihood-specific).

\subsection{Bayesian inference for imaging applications}\label{sec:Bayesforimaging}

The identifiability issues of the aforementioned likelihoods imply that additional information is needed to reduce the uncertainty about $x$ and deliver accurate well-posed solutions. Within the Bayesian framework, regularization is introduced via the prior distribution over $x$, denoted as $p(x)$. This distribution promotes expected properties about $x$ (e.g. sparsity, piece-wise regularity or smoothness). Inferences are then based on the posterior distribution, given via Bayes' theorem \cite{robert2007bayesian} by 

$$\pi(x) \triangleq p(x|y) = \dfrac{p(y|x)p(x)}{\int_{\mathbb{R}^{N}}p(y|x)p(x)dx}~~.$$

\noindent In this work, we focus on log-concave models of the form
\begin{equation}\label{post}
\pi(x) = \dfrac{e^{-f_{y}(x)-g(x)-\iota_{\mathbb{R}_{++}^{n}}(x)}}{\int_{\mathbb{R}^{d}}e^{-f_{y}(x)-g(x)-\iota_{\mathbb{R}_{++}^{n}}(x)}dx}~~,
\end{equation}
\noindent where $f_{y}:\mathbb{R}^{n} \rightarrow \mathbb{R}$ (corresponding to the data log-likelihood) and $g:\mathbb{R}^{n} \rightarrow (-\infty , +\infty]$ are two lower bounded functions and $\iota_{\mathbb{R}_{++}^{n}}(\cdot)$ is the indicator function on $\mathbb{R}_{++}^{n}$ which is convex and guarantees that the positivity constraint is satisfied. The two functions $f_{y}$ and $g$ satisfy the following conditions:

\begin{enumerate}
    \vspace{0.2cm}
    \item $f_{y}$ is convex and can be expressed as $f_{y}(x)=F(Ax)$ for some $F:\Rint\to \R$ which is Lipschitz continuous on the set $\{x\in \Rint: x_i>b \text{ for all } i\}$ for any $b>0$.
    \item $g$ is proper, convex and lower semi-continuous, but potentially non-smooth.
    \vspace{0.2cm}
\end{enumerate}

Calculating $\pi$ exactly is often not feasible in imaging because of dimensionality of $x$. Instead, summaries such as Bayesian estimators posterior probabilities or expectations are used to deliver some important information about $\pi$. In particular, for log-concave models as \eqref{post}, recent works often rely on the MAP estimator \cite{pereyra2016maximumaposteriori,survey}
\begin{equation}\label{MAP_eq}
\hat{x}_{MAP} = \argmax_{x} \pi(x) = \argmin_{x}\{f_{y}(x) + g(x) + \iota_{\mathbb{R}_{++}^{n}}(x)\}\,\,\,\,.
\end{equation}
\noindent which can efficiently computed, even in high dimensions, by using proximal convex optimization techniques \cite{chambolle_pock_2016,combettes2010proximal}. For instance, to solve the minimization problem of the form \eqref{MAP_eq} under a Poisson likelihood, an ADMM variant was considered in \cite{5492199} which does not assume Lipschitz continuity in the gradient while a quadratic approximation of the Poisson likelihood that turns the original problem into a constrained $l_{2}$ denoising problem was proposed in \cite{SPIRAL_paper}. 

Log-concavity also plays a central role in the calculation of moment-based estimators such as the minimum mean square error (MMSE) estimator since it guarantees the existence of all posterior moments. It also enables the efficient calculation of other advanced quantities related to uncertainty quantification, calibration of model parameters, and model selection in the absence of ground truth \cite{2018,durmus2016efficient,pereyra2016maximumaposteriori,vidal2020maximum}. More precisely, to calculate such estimators or perform more advanced analyses, it is necessary to use high-dimensional MCMC methods to simulate samples from the posterior $\pi$ followed by Monte Carlo integration \cite{green2015bayesian}. In recent years, highly efficient MCMC methods have been developed to perform Bayesian inference for imaging sciences based on techniques from proximal optimization. These methods are known as proximal MCMC methods \cite{cheng2018underdamped,durmus2016efficient,goldman2021gradientbased,pereyra2015proximal,vargas2020accelerating,8683031} and are useful for Bayesian imaging since they are easy to implement, have significantly reduced computational time (compared to traditional Gibbs and Metropolis-Hastings samplers) and offer theoretical convergence guarantees. However, they have been implemented for Gaussian imaging problems and cannot directly be applied to low-photon imaging problems that involve poor
regularity properties and non-negativity constraints. In this work, we will present how proximal MCMC methods can be adjusted to tackle the  challenging class of models \eqref{post}.

\section{Langevin SDEs for sampling}\label{section_3}

Before presenting the proposed reflected SDE we first recall how to use the Langevin SDE for sampling from distributions of the form
\begin{equation}\label{pi_U}
    \pi(x)=\frac{e^{-U(x)}}{\int_{\mathbb{R}^N} e^{-U(x)} dx}\,\,,
\end{equation}
for some function $U:\mathbb{R}^n\to\mathbb{R}$ with $e^{-U}$ integrable. If $\pi$ is given by \eqref{post} then $U(x)=f_{y}(x)+g(x)+\iota_{\mathbb{R}_{++}^n}(x)$. Before we discuss the difficulties in such a model, we first consider the case where $U$ is Lipschitz continuously differentiable on $\mathbb{R}^n$. As is shown in \cite{ma2019irreversible} the class of SDEs which target $\pi$ is given by
\begin{equation}\label{eq:generalSDE}
    \begin{aligned}
    dX_t&=[-(S(X_t)+J(X_t))\nabla_x U(X_t)+\Gamma(X_t)]dt+\sqrt{2S(X_t)}dW_t\,\,,\\
    \Gamma_i(x)&=\sum_{j=1}^n \partial_{x_j}(S_{ij}(x)+J_{ij}(x))\,\,,
    \end{aligned}
\end{equation}
for some  positive semidefinite diffusion matrix $S(x)$ and  skew-symmetric matrix $J(x)$. Here, $W_t$ is a $n$-dimensional Brownian motion. We concentrate on the overdamped Langevin SDE which corresponds to taking $J=0$ and $S(x)=I_n$, that is 
\begin{equation}\label{eq:Langevin}
    \begin{aligned}
    dX_t&=-\nabla_x U(X_t)dt+\sqrt{2}dW_t\,\,.
    \end{aligned}
\end{equation}
Another choice which is popular for sampling algorithms is the underdamped Langevin SDE which corresponds to taking $X_t=(q_t,p_t)$, $U(q,p)=U(q)+\lvert p\rvert^2/(2u)$, $$S=\left(\begin{array}{cc}0 & 0\\ 0& \gamma u\end{array}\right), \quad  J=\left(\begin{array}{cc}0 & -u\\ u& 0\end{array}\right),$$ for some constants $u,\gamma>0$, that is
\begin{equation}\label{eq:underdampedLangevin}
    \begin{aligned}
    dq_t&=p_tdt\\
    dp_t&=-u\nabla_x U(q_t)dt-\gamma p_tdt+\sqrt{2\gamma u}dW_t\,\,.
    \end{aligned}
\end{equation}

\subsection{Sampling in a domain}\label{sec:samplingdomain}

In the case where $U(x)=f_{y}(x)+g(x)+\iota_{\mathbb{R}_+^n}(x)$ we no longer have that the SDE \eqref{eq:Langevin} or even the more general form of SDE \eqref{eq:generalSDE} admits $\pi$ as an invariant measure. Indeed there are several issues: (i) the function $g$ does not need to be differentiable, so $\nabla_x g$ is not well-defined; (ii) $\pi$ is only positive on a subset of the domain but the SDE \eqref{eq:Langevin} is irreducible so any invariant measure of \eqref{eq:Langevin} admits a positive density on $\mathbb{R}^n$; (iii) the function $\nabla_x f$ is not Lipschitz continuous and therefore the SDE may not be well-posed. Let us discuss each of these issues in turn. 

\paragraph{$\nabla g$ is not well defined.} Since $g$ is not differentiable, we approximate $g$ by a function $g^\lambda$ which is obtained by applying the Moreau-Yosida (MY) envelope on g, that is,
\begin{equation}\label{MY_envelope}
    g^\lambda(x)=\min_{y\in \mathbb{R}^n}\{g(y)+\frac{1}{2\lambda} \lVert y-x\rVert^2\}
\end{equation}
where $g^{\lambda}\rightarrow g$ as $\lambda\rightarrow 0$. It is shown in \cite[Proposition 12.19]{rockafellar2009variational} that $g^\lambda$ is Lipschitz continuously differentiable with Lipschitz constant $\lambda^{-1}$ and gradient
\begin{equation}\label{MY_grad}
    \nabla g^\lambda(x)=\dfrac{1}{\lambda}\left(x-\prox_{g}^{\lambda}(x)\right)\,\,,
\end{equation}
\noindent where
$$\prox_{g}^{\lambda}(x) = \argmin_{y\in \mathbb{R}^n}\{g(y)+\frac{1}{2\lambda} \lVert y-x\rVert^2\}.$$
\noindent The use of \eqref{MY_grad} to approximate terms such as the function $g$ is investigated in \cite{durmus2016efficient}.

\paragraph{Positivity constraint.} Since the SDE \eqref{eq:Langevin} is non-degenerate, if it admits an invariant measure, then this measure must have a strictly positive density. Therefore we can not use the SDE \eqref{eq:Langevin} to sample from a measure which is supported on $\mathbb{R}_+^n$. In principle, by taking the matrix $S(x)$ to be degenerate one can construct an SDE which has an invariant measure supported on the manifold $\mathbb{R}_+^n$, however this introduces further complications such as not having a unique invariant measure (see \cite{CCDO} for a discussion of the invariant measures of degenerate SDEs). Therefore, we propose to introduce reflections on the boundary of $\mathbb{R}_+^n$ and construct a reflected SDE (RSDE) which targets an invariant measure on $\mathbb{R}_+^n$. 

\paragraph{$\nabla f$ not well-defined.} The function $\nabla_xf_{y}$ is not necessarily Lipschitz-continuous and is not necessarily finite on the boundary of $\mathbb{R}_+^n$. Thus, it is not clear that the SDE \eqref{eq:Langevin} is well-defined on $\mathbb{R}_+^n$ (even assuming at this stage that reflections are included in a suitable manner and that $g$ is replaced by $g^\lambda$). Even if the SDE was well-defined, these remain challenging for its numerical approximation since commonly used discrete time approximations need not be stable in the absence of a Lipschitz condition (see \cite{mattingly2002ergodicity}). To address this difficulty, we propose to introduce the approximation $f_{y}^b$ defined by $f_{y}^b(x)=F(Ax+b)$ of $f_{y}$. Note that by the assumptions of $F$, the function  $\nabla_xf_{y}^b$ is globally Lipschitz continuous. For these reasons, we propose to approximate $\pi$ by the measure $\pi^{\lambda,b}$ which is supported on $\mathbb{R}_+^n$ and defined for $x\in \R_+^n$ by
\begin{align}\label{eq:pilb}
    \pi^{\lambda,b}(x)=\frac{e^{-\Ulb(x)}}{\int_{\R_+^n} e^{-\Ulb(\tilde{x})} d\tilde{x}}, \quad 
    \Ulb(x)=f_{y}^b(x)+g^\lambda(x).
\end{align}
In the following lemma we show that $\pi^{\lambda,b}$ is an approximation of $\pi$ that has interesting asymptotic properties as $\lambda$ and $b$ vanish.
\begin{lemma}\label{lem:TVconvergence}
Suppose that the conditions of Section \ref{sec:Bayesforimaging} hold. Let $f_y^b$ and $g^\lambda$ be as above then with $\pi^{\lambda,b}$ given by \eqref{eq:pilb} and $\pi$ defined by \eqref{post}, we have
\begin{equation*}
    \lim_{\lambda\to 0, b\to 0}\lVert \pi-\pi^{\lambda,b}\rVert_{TV} =0.
\end{equation*}
\end{lemma}

\begin{proof}
    The proof is deferred to Appendix \ref{app:proofs_exp_ergodicity}.
\end{proof}

\subsection{Reflected SDE}\label{section_ref_SDE}

In Section \ref{sec:samplingdomain} we discussed the issue of sampling from the distribution $\pi$ given by \eqref{post} and introduced an approximation $\pi^{\lambda,b}$ defined on the domain $\R_+^n$. In this section we introduce a RSDE which has $\pi^{\lambda,b}$ as an invariant distribution and discuss its properties. In Section \ref{sec:expergodicity} we verify that this RSDE has the desired invariant distribution and converges exponentially.

We consider the following RSDE defined on the convex domain $\R_+^n\subseteq \R^n$ 
\begin{align}\label{eq:RSDE}
dX_t = -\nabla_x\Ulb(X_t)dt+\sqrt{2}dW_t + d\kappa_t.
\end{align}
here $W_t$ is a $n$-dimensional Brownian motion, $\kappa_t$ is local time which only increases on $\partial \R_+^n$, $\Ulb:\R_+^n\to\R$ and is $C^2$ with $\nabla_x\Ulb$ being globally Lipschitz continuous. Observe that this RSDE is non-negative as it takes values in $\R^n_+$.  We define the semigroup, $\cP_t$, corresponding to this RSDE by
\begin{equation*}
    \cP_t\varphi(x) = \mathbb{E}[\varphi(X_t^x)] \quad \text{ for all } \varphi\in B_b(\R_+^n), x\in \R_+^n, t\geq 0.
\end{equation*}
For background on semigroups and operators we refer the reader to \cite{Ethier}. We (formally) define the generator corresponding to this semigroup to be
\begin{equation*}
    \cL \varphi(x) = -\langle\nabla_x\Ulb(x), \nabla_x\varphi(x)\rangle +\Delta_x\varphi(x)
\end{equation*}
for $\varphi:\R_+^n\to \R$ sufficiently smooth. Under the above assumptions it is shown in \cite[Theorem 4.1]{TanakaRSDE} there is a (pathwise) unique strong solution to \eqref{eq:RSDE}. 
We emphasise that the domain does not need to be smooth provided that it is convex, in particular the domain $\R_+^n$ is permitted. Let us expand on what we mean by local time in this context, $\kappa:[0,\infty)\to \R^n$ is a c\'adl\'ag function with bounded variation, $\kappa_0=0$, the set $\{t\geq 0: X_t^i>0 \text{ for all } i\}$ has measure zero with respect to the measure $d\lvert \kappa\rvert$ and we can write
\begin{equation*}
    \kappa_t=\int_0^t\hat{n}(s)d\lvert \kappa_s\rvert
\end{equation*}
where the function $\hat{n}(s)$ is a unit normal vector at $X_s$ for almost all $s$ with respect to the measure $d\lvert \kappa\rvert$. Since the domain needs not to be smooth we shall explain what we mean by a normal vector. A normal vector $\hat{n}$ at $x$ is any unit vector which is orthogonal to some supporting hyperplane $H$ of $\R_+^n$. A supporting hyperplane of $\R_+^n$ is a hyperplane such that $\R_+^n$ is contained in one of the two closed half-spaces bounded by $H$ and $\R_+^n$ has at least one boundary point on the hyperplane. In this setting, the set of (inward) normal vectors at $0$ is the set
$$
\{\hat{n}=(\hat{n}_1,\ldots,\hat{n}_n): \lvert \hat{n}\rvert=1, n_i\geq 0 \, \text{ for all } i\}.
$$
If the set $\R_+^n$ had a smooth boundary and the coefficients are suitably smooth (see \cite{freidlin} for example), it is well known that the semigroup $\cP_t\varphi$ corresponding to $X_t$ solves a Neumann boundary value problem.
As the case we are interested in has a non-smooth boundary for $n>1$ more care is required. This domain is investigated in \cite{DEUSCHEL} under the assumption that the drift coefficient $\nabla_x\Ulb\in C^1(\R_+^n)$ is globally Lipschitz, and it is shown that the semigroup $\cP_t\varphi\in C^1(\R_+^n)$ for $\varphi\in C_b(\R_+^n)$ and satisfies the Neumann boundary conditions $\partial_{x_i}\cP_t\varphi(x)=0$ whenever $x_i=0$. Therefore the semigroup $u_t(x)=\cP_t\varphi$ solves the Neumann boundary value problem:
\begin{equation}\label{eq:NPDE}
    \left\{\begin{aligned}
    \partial_tu_t(x)&=\cL u_t(x),   && x\in \R_+^n,\\
    \frac{\partial u_t}{\partial x_i}(x)&=0,     && \text{ for any  } x\in \R_+^n \text{ with } x_i=0,\\
    u_0(x)&=\varphi(x), && x\in \R_+^n.
    \end{aligned}\right.
\end{equation}

In Proposition \ref{prop:smoothing} (see Appendix \ref{app:proofs_exp_ergodicity}) we apply the results of \cite{DEUSCHEL} to describe the smoothness of the semigroup, indeed for any $\varphi\in C_b(\R_+^n)$ we have that $\cP_t\varphi\in C^1(\R_+^n)$ and by ellipticity we have $\cP_t\varphi\in C^\infty(\Rint)$. 


\subsection{Exponential Ergodicity of RSDEs}\label{sec:expergodicity}

In order to show exponential ergocidity we shall apply \cite[Theorem 6.1]{Meyn}. Before giving the details of this theorem we provide an informal summary of the main assumption, the existence of a Lyapunov function. For our purposes, a Lyapunov function is a function $V:\R_+^n\to \R_+$ which is positive, $C^2$, $V(x)\to \infty$ as $x\to \infty$ and satisfies for some $C_1>0, C_2\in \R$
\begin{align}
    &\cL V\leq -C_1 V + C_2\label{eq:Lyapunovconditionsimple}\\
    &\frac{\partial V}{\partial x_i}(x)  =0 \quad \text{ whenever } x_i=0, \text{ for some } i\in \{1,\ldots, n\}.\label{eq:Lyapunovconditionboundary}
\end{align}
The first condition \eqref{eq:Lyapunovconditionsimple} is a sufficient condition to have that $\cP_tV$ is bounded in $t$ hence that the process is non-explosive and can not spend too much time in the tails. The second condition \eqref{eq:Lyapunovconditionboundary} ensures that the boundary condition arising from the non-negativity constraint is satisfied. This condition can be viewed as a necessary condition to have that $V$ is in the domain $D(L)$ of the generator. However as $V$ is unbounded it is not clear that $V$ belongs to the domain of the generator, indeed we do not even have that $\cP_t V$ is well-defined at this stage. For this reason we introduce a sequence of operators $\cL_\ell$ defined on functions from some bounded set $O_\ell$ to $\R$. We construct such functions by considering the process $X_{t\wedge T^\ell}$ where $T^\ell$ is the first exit time of $O_\ell$, i.e. the process is equal to $X_t$  up until the first exit time of $O_\ell$ and then remains constant. Set $\cL_\ell$ to the generator of this process. Since $O_\ell$ is bounded we have that $V\in C_b(O_\ell)$ and that $V$ belongs to the domain of $\cL_\ell$. We now give the more precise statement of this assumption, which is phrased in terms of the extended generator. 
We assume there is a sequence of open (as a subset of $R_+^n$) and bounded sets $O_\ell$ such that $O_\ell\subseteq O_{\ell+1}$ for all $\ell\geq 1$ and $\bigcup_{\ell=1}^\infty O_{\ell}=\R^n_+$. Let $T^\ell $ be the first exit time of $O_\ell$ and set $\cL_\ell$ to be the extended generator of the process $\{X_{t\wedge T^\ell}^x\}_{t\geq 0}$, that is a function $\varphi:\R_+^n\times \R_+\to \R$ is in the domain of $\cL_\ell$ if there exists a function $\psi:\R_+^n\times \R_+\to \R$ such that for each $x\in \R_+^n$, $t>0$
\begin{align*}
    \mathbb{E}[\varphi(X_{t\wedge T^\ell}^x,t)]=\varphi(x,0)+\mathbb{E}\left[\int_0^t \psi(X_{s\wedge T^\ell}^x,s)ds\right], \quad 
    \mathbb{E}\left[\int_0^t \lvert \psi(X_{s\wedge T^\ell}^x,s)\rvert ds\right]<\infty,
\end{align*}
and we write $\cL_\ell \varphi=\psi$. In our setting we can set $O_\ell=[0,\ell)^n$ and the domain of $\cL_\ell$ contains all $C^2$ functions with $\partial_{x_i}f(x)=0$ whenever $x_i=0$ by It\^o's formula. Now we shall state the conditions we require to apply \cite[Theorem 6.1]{Meyn} which is \cite[(CD3)]{Meyn} restated in our notation.

\begin{hypothesis}\label{hyp:Lyapunov}
There exists a function $V:\R_+^n\to \R_+$ which is positive, measureable, $V(x)\to \infty$ as $x\to\infty$ and for some $C_1>0, C_2<\infty$
\begin{equation}\label{eq:Lyapunovcondition}
    \cL_\ell V(x) \leq -C_1 V(x) +C_2, \quad x\in O_\ell.
\end{equation}
\end{hypothesis}

\begin{theorem} (\cite[Theorem 6.1]{Meyn}) \label{thm:MeynTweedie}
    Suppose that $X_t$ is a right process, and that all compact sets are petite (see \cite{Meyn} for definitions). If Hypothesis \ref{hyp:Lyapunov} holds, then $X_t^x$ admits a unique invariant measure $\pi$ and there exists $\beta<1$ and $B<\infty$ such that
    \begin{equation*}
        \sup_{\lvert \varphi\rvert \leq 1+V}\lvert \cP_t\varphi(x) -\pi(\varphi)\rvert \leq B(1+V(x))\beta^t, \quad t\geq 0, x\in \R_+^n.
    \end{equation*}
\end{theorem}

We can now apply this theorem for the RSDE \eqref{eq:RSDE}.

\begin{theorem}\label{thm:expergodicityconvex}
Let $\cP_t$ be the semigroup corresponding to the RSDE \eqref{eq:RSDE}. Assume that $\nabla_x\Ulb$ is continuously differentiable and globally Lipschitz. Suppose that $\Ulb$ is convex and $e^{-\Ulb(x)}$ is integrable over $\R_+^n$ then $X_t^x$ admits a unique invariant measure $\pi$ and there exist $V:\R_+^n\to [0,\infty)$, $\beta<1$ and $B<\infty$ such that
\begin{equation*}
    \sup_{\lvert \varphi\rvert \leq 1+V(x)}\lvert \cP_t\varphi(x) -\pi(\varphi)\rvert \leq B(1+V(x))\beta^t, \quad t\geq 0, x\in \R_+^n.
\end{equation*}
Moreover, the invariant measure $\pi$ is given by \eqref{eq:pilb}  for any $x\in \R_+^n$.
\end{theorem}

\begin{proof}
    The proof is deferred to Appendix \ref{app:proofs_exp_ergodicity}.
\end{proof}
\begin{remark}
Since all the models considered in Section \ref{sec:observationmodels} correspond to convex negative log-likelihoods, by imposing a log concave prior which is integrable we have that the assumptions of Theorem \ref{thm:expergodicityconvex} are satisfied and hence the corresponding RSDE is exponentially ergodic in each of these cases. Note that the proof of the above theorem does not explicitly require that $\Ulb$ is convex but only that there exist $\alpha>0$ and $R>0$ such that for $\lvert x\rvert \geq R$,
\begin{equation}\label{eq:convexlowerbound}
    \langle x, \nabla \Ulb(x)\rangle \geq \alpha \lvert x\rvert.
\end{equation}
It is shown in \cite[Lemma 2.2]{Bakrysimpleproof} that any convex $C^1$-function $\Ulb$ which is integrable must satisfy \eqref{eq:convexlowerbound}, however as \eqref{eq:convexlowerbound} only needs to hold outside a compact set it is clear that there are also non-convex functions $\Ulb$ for which Theorem \ref{thm:expergodicityconvex} holds.
\end{remark}

\subsection{Discrete time approximations of RSDEs} \label{discrete_rsde}

In Section \ref{section_ref_SDE} we introduced an RSDE which is exponentially ergodic and has the measure $\pi$ as an invariant distribution however this leads to the question of how do we approximate this process. In this section we discuss different approaches to approximate the local time term in the RSDE. There are three approximations we will discuss: (i) Penalty scheme; (ii) Projection scheme; (iii) Reflection scheme. Note there are more complicated schemes such as the half-plane approximation (see \cite{Gobet2001} for details) however, as this is more expensive and leads to the same order of accuracy as the reflection scheme, we do not consider it. 

The penalty scheme replaces the local time term in \eqref{eq:RSDE} by $\beta_\varepsilon(X_t)dt$ where $    \beta_\varepsilon(x) = (x^+-x)/\varepsilon$
and $x^+$ is the vector whose $i$-component is $\max(x_i,0)$. Then for each fixed $\varepsilon$ we have an SDE defined on the whole space $\R^n$. An alternative way to view this scheme is to consider the measure $\pi$ given by \eqref{post} and apply the proximal operator to the indicator function $\iota_{\R^n_{++}}$ in the same way as we approximate $g$. This results in a measure $\pi^{\lambda,b,\varepsilon}$ supported on $\R^n$ which can then be viewed as the invariant measure of the Langevin SDE
\begin{equation*}
    dX_t=-\nabla_xf^b(X_t)dt-\nabla_xg^\lambda(X_t)dt+\beta_\varepsilon(X_t)dt+\sqrt{2}dW_t.
\end{equation*}
This SDE can then be approximated by standard numerical schemes for SDEs such as the Euler-Maruyama scheme. In order to have a scheme which converges to the RSDE as the time step tends to zero we need to let $\varepsilon$ be a function of the time step $h$. A potential issue with this approach is that it allows for $X_t$ to be negative which can lead to numerical instabilities. For a further discussion about the convergence of this scheme in terms of mean square error, the reader is invited to consult \cite{petterson}.

In contrast to the penalty scheme, the projection scheme ensures that the samples are always in $\R^n_+$ by applying a projection at the end of each step. This scheme is defined by 
\begin{align*}
Y_{t_{k+1}}&=\bar{X}_{t_{k}}-\nabla_xU(X_{t_k}) h +\sqrt{2h}\xi_k, \\
\bar{X}_{t_{k+1}}^i &= (Y_{t_{k+1}}^i)^+ \quad \text{ for all } i\in \{1,\ldots, n\}.
\end{align*}
Here $\pi_D^\gamma(x) = \pi_{\partial D}^\gamma(x)+(F(x))^+\gamma(\pi_{\partial D}^\gamma(x))$, $t_k=kh$, and $\xi_k$ are i.i.d random variables with mean zero, covariance given by the identity matrix and finite third moments. It has been shown in \cite{Costantini} that the weak error converges with order $h^{\frac{1}{2}-\epsilon}$ for any $\epsilon>0$ and moreover that the rate $h^{\frac{1}{2}}$ is a lower bound for reflecting Brownian motion in an interval. In the case where $\xi_k$ is bounded, the weak error converges with order $h^{\frac{1}{2}}$ \cite[Theorem 3.4]{Costantini}.

The reflected Euler scheme is similar to the projected Euler scheme but the process is in the interior of the domain with probability one, i.e. $\bar{X}_{t_k}^{i}\neq 0$ for any $i,k$ with probability $1$. This scheme is defined by 
\begin{align*}
Y_{t_{k+1}}&=\bar{X}_{t_{k}}-\nabla_x\Ulb(X_{t_k}) h +\sqrt{2h}\xi_k, \\
\bar{X}_{t_{k+1}} &= \lvert Y_{t_{k+1}}\rvert \quad \text{ for all } i\in \{1,\ldots, n\}.
\end{align*}
Here $\lvert \cdot \rvert$ is understood to be applied componentwise, i.e. for $x\in \R^n$ we define the vector $\lvert x\rvert$ by $(\lvert x\rvert)_i=\lvert x_i\rvert$. It is shown in \cite[Theorem 1]{bossy_gobet_talay_2004} that this scheme has weak error with order $h$, comparing favourably with the projection scheme. 

From the options that we have discussed above, we choose to adopt the reflected scheme as a way of incorporating the non-negativity constraint without increasing the computational complexity of the methods or introducing significant additional bias. The resulting proximal Markov chain Monte Carlo algorithms are presented in Section \ref{section4}. From an algorithmic viewpoint, these MCMC algorithms are closely related to their non-reflected counterparts, with the only minor change being the careful introduction of a component-wise absolute-value operation. However, from a Bayesian computation methodology perspective, this minor change has the profound effect of allowing us to approximate the reflected SDE, which we have shown to be well-posed and to converge exponentially fast to the desired target density.

\subsection{Related work}

Let us mention some related works on the use of RSDEs for sampling. In the case of a smooth bounded domain, convergence of RSDEs to their invariant distribution and numerical approximations as studied in \cite{cattiaux2017invariant,leimkuhler2020simplest}. In both of these papers the invariant distribution is unknown (as they work with general RSDEs) but can be shown to exist and be unique since the domain is bounded and the noise is non-degenerate. Convergence of the numerical scheme is then investigated by using suitable expansions and PDE estimates. In the case of convex bounded domains the papers \cite{bubeck2018sampling,lamperski2021projected} establishes exponential convergence of the RSDE to the invariant distribution and convergence of the projected Euler scheme. Reflected Langevin dynamics have also been used in the context of optimisation, see \cite{sato2022convergence}, in the case of a bounded smooth domain. Finally we mention that reflections can also be used with Hamiltonian dynamics \cite{chalkis2021truncated,dobson2020reversible}, although we note here that a Metropolis-Hastings step has been included to ensure the algorithm has the desired invariant distribution. Since all of these algorithms are for bounded domains they are not applicable to our setting and extra care is required as the domain is not smooth.

\section{Reflected MCMC methods}\label{section4}

In general, it is not possible to solve \eqref{eq:RSDE}, and discrete approximations of the overdamped and underdamped Langevin dynamics (Section \ref{section_3}) need to be considered instead. In this Section, numerical schemes will be described that aim to solve \eqref{eq:RSDE} in a high dimensional setting. These schemes are inspired by the Reflected Euler scheme provided in Section \ref{discrete_rsde} and adjust current proximal MCMC methodology \cite{cheng2018underdamped,durmus2016efficient,laumont2021bayesian,vargas2020accelerating} that has proven to be scalable and efficient under Gaussian noise but cannot guarantee that the non-negative constraint is satisfied under Poisson, binomial or geometric noise processes.

\subsection{Reflected MYULA (PMALA)}\label{sec:RMYULA}

In the absence of the non-negativity constraint requirement, a simple way to discretize the SDE in (\ref{eq:Langevin}) is to follow the Euler-Maruyama (EM) scheme leading to the known Unadjusted Langevin Algorithm (ULA) \cite{roberts1996exponential}. Under a non-differentiable $U = -\log\pi$, the authors in \cite{pereyra2015proximal} approximated the non-smooth term $g$ of $U$ by its MY envelope $g^{\lambda}$ given in \eqref{MY_envelope}. By using \eqref{MY_grad}, ULA could be applied on (\ref{eq:Langevin}) leading to the known Moreau-Yosida Unadjusted Langevin Algorithm (MYULA). Taking into consideration the non-negativity constraint requirement, we propose a reflected version of MYULA, namely Reflected MYULA (R-MYULA), where each sample is simply reflected by taking its absolute value (see Section \ref{discrete_rsde} ). The algorithm is presented in \eqref{alg:r_MYULA}. If necessary, the asymptotic bias introduced by the discretization of the RSDE in \eqref{eq:RSDE} and the approximation $\pi^{\lambda,\beta}$ of $\pi$ can be corrected by complementing the R-MYULA with a Metropolis-Hastings (MH) step \cite{durmus2016efficient,pereyra2015proximal,roberts1996exponential} leading to the Reflected Proximal Metropolis Adjusted Langevin Algorithm (R-PMALA). 

A main computational drawback of R-MYULA and R-PMALA is the fact that in order to converge one needs to choose $\delta \leq 1/L$ where $L=L_{f^{b}_{y}}+1/\lambda$ is the Lipschitz constant of $\nabla \log\pi^{\lambda,b}$ and $L_{f^{b}_{y}}$ is the Lipschitz constant of $f^{b}_{y}$ constituting a function of $b$. The parameters $\lambda$ and $b$ define a trade-off between the computational efficiency and accuracy since letting $\lambda\rightarrow 0$ and $b\rightarrow 0$ brings $\pi^{\lambda,b}$ close to $\pi$, which reduces the asymptotic bias, but leads the discretization time-step to be diminished and consequently reduces the chain efficiency \cite{durmus2016efficient}. Thus, in cases where the likelihood is very informative ($L_{f^{b}_{y}}$ is large) or a high level of accuracy in the approximated $g^{\lambda}$ is required, the stepsize will be small leading to a highly correlated Markov chain and an MCMC method that explores very slowly the target distribution $\pi^{\lambda,b}$.

\begin{algorithm}[h]
\caption{Reflected MYULA}
\label{alg:r_MYULA}
\begin{algorithmic}
\STATE{\textbf{Set} $X_{0}\in \mathbb{R}^{n}_{++}$, $b>0$, $\lambda>0$, $\delta\in(0,\lambda/(\lambda L_{f_{y}^{b}}+1)]$, $N\in \mathbb{N}$}
    \FOR{$k = 0:(N-1)$}
\STATE{$Z_{k+1} \sim \mathcal{N}(0,\mathbb{I}_{n})$}
\vspace{0.1cm}
\STATE{$Y_{k+1} = \left(1-\frac{\delta}{\lambda}\right)X_{k}-\delta\nabla f^{b}_{y}(X_{k}) + \frac{\delta}{\lambda}\prox_{g}^{\lambda}(X_{k}) +\sqrt{2\delta}Z_{k+1}$ (MYULA)}
\vspace{0.1cm}
\STATE{$X_{k+1} = |Y_{k+1}|$}
\ENDFOR
    \RETURN $\{X_{k}: k\in\{1,\ldots, N\}\}$
\end{algorithmic}
\end{algorithm}

\begin{algorithm}[h!]
\caption{Reflected SKROCK}
\label{alg:r_SKROCK}
\begin{algorithmic}
\STATE{\textbf{Set} $X_{0}\in \mathbb{R}^{n}_{++}$, $b>0$, $\lambda>0$, $N\in \mathbb{N}$, $s\in\{3,\ldots,15\}$, $\eta>0$}
\STATE{\textbf{Compute} $l_{s}=(s-0.5)^{2}(2-4/3\eta)-1.5\vspace{0.2cm}$}
\STATE{\textbf{Compute} $\omega_{0} = 1+\dfrac{\eta}{s^{2}},\hspace{0.35cm} \omega_{1} = \dfrac{T_{s}(\omega_{0})}{T^{'}_{s}(\omega_{0})},\hspace{0.35cm} \mu_{1}  = \dfrac{\omega_{1}}{\omega_{0}},\hspace{0.35cm} \nu_{1}=s\omega_{1}/2,\hspace{0.35cm} k_{1} = s\omega_{1}/\omega_{0} \vspace{0.2cm}$}
\STATE{\textbf{Choose} $\delta\in(0,\delta_{s}^{max}]$, where $\delta_{s}^{max}= l_{s}/(L_{f_{y}^{b}}+1/\lambda)$}
\FOR{$k = 0:(N-1)$}
\STATE{$Z_{k+1} \sim \mathcal{N}(0,\mathbb{I}_{n})$}
\vspace{0.1cm}
\STATE{$K_{0} = X_{k}$}
\STATE{$W_{1} = X_{k} + \nu_{1}\sqrt{2\delta}Z_{k+1}$}
\STATE{$Y_{1} = X_{k} - \mu_{1}\delta\nabla f^{b}_{y}(W_{1}) -  \frac{\mu_{1}\delta}{\lambda}(W_{1}-\prox_{g}^{\lambda}(W_{1}))+ k_{1}\sqrt{2\delta}Z_{k+1}$}
\STATE{$K_{1} = |Y_{1}|$}
\vspace{0.1cm}
\FOR{$j = 2:s$}
\STATE{\textbf{Compute}}
$\mu_{j} = \dfrac{2\omega_{1}T_{j-1}(\omega_{0})}{T_{j}(\omega_{0})},\hspace{0.35cm} \nu_{j}=\dfrac{2\omega_{0}T_{j-1}(\omega_{0})}{T_{j}(\omega_{0})},\hspace{0.35cm} k_{j} = -\dfrac{T_{j-2}(\omega_{0})}{T_{j}(\omega_{0})}= 1-\nu_{j}\vspace{0.2cm}$
\STATE{$Y_{j} = -\mu_{j}\delta\nabla f^{b}_{y}(K_{j-1}) - \frac{\mu_{j}\delta}{\lambda}(K_{j-1}-\prox_{g}^{\lambda}(K_{j-1})) + \nu_{j}K_{j-1}+k_{j}K_{j-2}$}
\STATE{$K_{j} = |Y_{j}|$}
\ENDFOR
\STATE{$X_{k+1} = K_{s}$}
\ENDFOR
\RETURN $\{X_{k}: k\in\{1,\ldots,N\}\}$
\end{algorithmic}
\end{algorithm}

\begin{algorithm}[h!]
\caption{Reflected MYUULA}
\label{alg:r_MYUULA}
\begin{algorithmic}
\STATE{\textbf{Set} $X_{0}\in \mathbb{R}^{n}_{++}$, $V_{0}\in \mathbb{R}^{n}$, $b>0$, $\lambda>0$, $\gamma>0$, $u = \lambda/(\lambda L_{f^{b}_{y}}+1)$, $\delta = O(1)$, $N\in \mathbb{N}$}
\vspace{0.2cm}
\STATE{\textbf{Set} 
$\Sigma = 
\begin{pmatrix}
u\left(\delta-\frac{1}{\gamma^{2}}e^{-2\gamma\delta}-\frac{3}{\gamma^{2}}+e^{-\gamma\delta}\right)\mathbb{I}_{n} & \frac{u}{\gamma}(1+e^{-2\gamma\delta}-2e^{-2\gamma\delta})\mathbb{I}_{n} \\
\frac{u}{\gamma}(1+e^{-2\gamma\delta}-2e^{-2\gamma\delta})\mathbb{I}_{n} & u(1-e^{-2\gamma\delta})\mathbb{I}_{n}
\end{pmatrix}
$
}
\vspace{0.2cm}
\FOR{$k = 0:(N-1)$}
\vspace{0.1cm}
\STATE{$(Z^{1}_{k+1}, Z^{2}_{k+1}) \sim \mathcal{N}(0,\Sigma)$}
\vspace{0.1cm}
\STATE{$V_{k+1} = V_{k}e^{-\gamma\delta}-\dfrac{u}{\gamma}(1-e^{-\gamma\delta})(-\nabla f^{b}_{y}(X_{k})-\frac{1}{\lambda}(X_{k}-\prox_{g}^{\lambda}(X_{k})))$}
\STATE{$Y_{k+1} = X_{k} + \dfrac{1}{\gamma}(1-e^{-\gamma\delta})V_{k} - \dfrac{u}{\gamma}\left(\delta - \dfrac{1}{\gamma}(1-e^{-\gamma\delta})\right)(-\nabla f^{b}_{y}(X_{k})-\frac{1}{\lambda}(X_{k}-\prox_{g}^{\lambda}(X_{k})))$}
\STATE{$X_{k+1}=|Y_{k+1}|$}
\ENDFOR
\RETURN $\{X_{k}: k\in\{1,\ldots,N\}\}$
\end{algorithmic}
\vspace{-0.1cm}
\end{algorithm}

\subsection{Reflected SKROCK}

The authors in \cite{vargas2020accelerating} dealt with the step-size limitation of
MYULA by applying a more sophisticated discretization scheme to simulate the Langevin SDE \eqref{eq:Langevin} based on explicit stabilized Runge-Kutta integrators \cite{abdulle2015explicit,optimal_explicit}. This scheme is known as the stochastic second kind orthogonal Runge-Kutta Chebyshev (SKROCK) method, and allows for accelerated and efficient sampling from the posterior distribution \cite{optimal_explicit,vargas2020accelerating}. Under the non-negativity requirement, we propose the Reflected SKROCK (R-SKROCK) algorithm by applying reflection after every gradient evaluation. The algorithm is presented in \eqref{alg:r_SKROCK}.

\subsection{Reflected MYUULA}

In a similar manner with MYULA, the authors in \cite{goldman2021gradientbased} used a MY smoothed approximation $U^{\lambda}$ of a non-differentiable $U=-\log\pi$ to solve the underdamped dynamics in \eqref{eq:underdampedLangevin}. Literature regarding formulations of Langevin-based algorithms to integrate (\ref{eq:underdampedLangevin}) can be found, for instance, in \cite{cheng2018underdamped,mclachlan_quispel_2002,sanzserna2021wasserstein}. In this work, we present a reflected version of the the well-known Euler exponential integrator \cite{cheng2018underdamped,hochbruck_ostermann_2010} to solve \eqref{eq:underdampedLangevin}. The algorithm is presented in \eqref{alg:r_MYUULA}.


\subsection{Implementation guidelines}\label{guidelines}

We now discuss suitable ranges and recommended values for the parameters of Algorithm \ref{alg:r_MYULA}, Algorithm \ref{alg:r_SKROCK} and Algorithm \ref{alg:r_MYUULA}. Our recommendations seek to provide general rules that are simple and robust, rather than optimal values for specific models.

\subsubsection*{Setting $\lambda$ and $b$} Similarly to \cite{durmus2016efficient}, we recommend to set $\lambda\in [L_{f^{b}_{y}}^{-1},10L_{f^{b}_{y}}^{-1}]$ and use $\lambda= L_{f^{b}_{y}}^{-1}$ in our numerical experiments (larger values of $\lambda$ lead to faster convergence at the cost of additional bias). Our experiments suggest that the value of $b$ is more critical in terms of controlling bias than $\lambda$. We recommend setting $b$ to be $1\%$ of the expected Mean Intensity Value (MIV) of $x$ (i.e., $||x||_{1}/n$), which provides a good trade-off of efficiency and accuracy. 

\subsubsection*{Setting $\delta$, $\gamma$, $u$, $s$, and $\eta$}
The convergence theory for MYULA \cite{durmus2016efficient} requires setting $\delta\in (0,\delta_{max})$ with $\delta_{max} = 1/(L_{f^{b}_{y}}+1/\lambda]$ (the value $\delta_{max}$ is derived from bounding the Lipschitz constant of $\nabla f_{y}^{b} + (x-\prox^{\lambda}_{g}(x))/\lambda$). In our experiments, given the dimensionality involved, and because we did not observe significant bias, we choose $\delta = \delta_{max}$ to improve convergence speed. Regarding R-MYUULA, following the spirit of \cite{cheng2018underdamped,sanzserna2021wasserstein}, we set $\delta = 2$, $u = 1/(L_{f^{b}_{y}}+1/\lambda)$, and $\gamma=2$. For R-SKROCK, similarly to \cite{vargas2020accelerating}, we recommend setting $\delta\in(0,\delta_{s}^{max}]$, where $\delta_{s}^{max}= l_{s}/(L_{f_{y}^{b}}+1/\lambda)$, $l_{s}=(s-0.5)^{2}(2-4/3\eta)-1.5\vspace{0.2cm}$ and $s\in\{3,\ldots,15\}$; we use $\delta = \delta_{s}^{max}$ with $s=10$ and $\eta=0.05$ in our experiments. Lastly, to maximise the asymptotic computational efficiency of R-PMALA we set $\delta$ to achieve an acceptance rate close to $57\%$ within the range of admissible values for geometric ergodicity $(0,\delta_{max})$ with $\delta_{max} = 1/(L_{f^{b}_{y}}+1/\lambda]$ \cite{optimal_scaling}.

\section{Numerical Experiments}\label{section_5}

In this section, the MCMC methods proposed in Section \ref{section4} are demonstrated through a range of experiments related to Poisson image deconvolution, binomial denoising and geometric image inpainting 
These experiments were selected to represent a variety of challenging configurations in terms of ill-posedness, ill-conditioning and dimensionality of $y$ and $x$. We report results for R-MYULA, R-SKROCK, R-MYUULA and R-PMALA and comparisons with MAP estimation by using the ADMM algorithm PIDAL \cite{5492199}. Because PIDAL was originally designed to work with Poisson likelihoods, in the binomial and geometric experiments we make minor changes to PIDAL so that it performs MAP inference w.r.t. to the correct model. Moreover, in the Poisson deconvolution experiment, we also report a comparison with the state-of-the-art Gibbs sampler SPA \cite{8683031}, which was successfully applied to total-variation Poisson image deblurring in \cite{8683031} (SPA combines MYULA with an augmentation-relaxation strategy to improve convergence speed, at the expense of more bias). 

To make the comparisons fair, and without loss of generality, we conduct all our experiments by using the same total-variation prior that is considered in PIDAL and SPA. This prior is given by 
$p(x|\theta)\propto \exp\{-\theta TV(x)\}$, where $TV(\cdot)$ denotes the isotropic total-variation pseudo-norm, and $\theta > 0$ is a regularisation parameter. We automatically estimate $\theta$ from $y$ for each experiment by marginal maximum likelihood estimation by using \cite[Algorithm 1]{vidal2020maximum}. Furthermore, for fairness, all methods are run with the same computing time budget. In particular, we generate $10^{6}$ samples for the R-MYULA and R-MYUULA, $10^{6}/s$ samples with $s=10$ for the R-SKROCK, $6.5 \times 10^{5}$ samples for the R-PMALA (there is an additional computational overhead associated with the Metropolis-Hastings (MH) step \cite{pereyra2015proximal}), and $6.8\times 10^{5}$ for SPA. In all cases, we use a $5\%$ burn-in period. Also, to check the asymptotic bias introduced by the approximation $\pi^{\lambda,b}$ of $\pi$, the discretization and the reflection, for each experiment we also run a long R-PMALA chain targetting $\pi$. In all cases, we observed a very good agreement between the unadjusted methods and the R-PMALA reference, indicating that the asymptotic bias is indeed very small.

The Monte Carlo samples produced by the MCMC methods are then used to compute the following quantities: 1) the posterior mean $E(x|y)$, which is the minimum mean square error (MMSE) Bayesian estimator of $x|y$; 2) uncertainty visualisation plots presenting the marginal standard deviation of pixels at different resolutions \cite{2018}; and 3) the sample autocorrelation function (ACF) of the fastest and slowest mixing components of the Markov chain (these correspond to the one-dimensional subspaces where the Markov chain achieves the highest and slowest convergence rates, corresponding to the subspaces with lowest and highest variance respectively\footnote{To identify the directions with smallest and largest uncertainty, someone would need to compute the posterior covariance matrix, an infeasible task in high dimensions. Instead, approximations of the posterior covariance are used by assuming that the latter is diagonalizable on the same basis as the forward operator $A$.}). For completeness, the ACF for a pixel with typical variance is also included in the plots (we use the median). Lastly, as a way of illustrating convergence speed, for each experiment we also plot the evolution of the estimated normalized root MSE (NRMSE) for the posterior mean estimate, as a function of the number of gradient evaluations.

\begin{figure}[h!]
    \centering
    \begin{subfigure}[t]{.24\linewidth}
        \centering
    	\includegraphics[width=\textwidth]{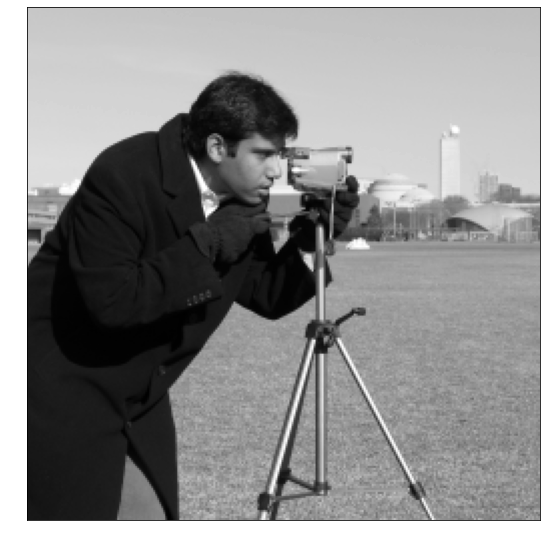}
    	\caption{\small{True image $x$}}
    	\label{fig:true_image_poisson}
    \end{subfigure} 
    \begin{subfigure}[t]{.24\linewidth}
        \centering
    	\includegraphics[width=\textwidth]{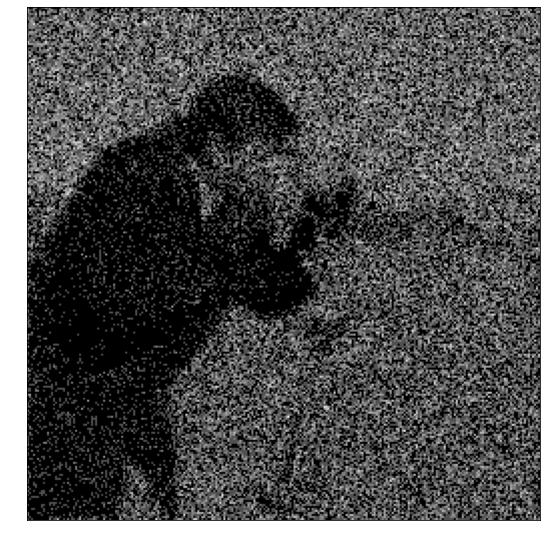}
    	\caption{\small{Observation $y$ (5.8 dB)}}
    	\label{fig:noisy_image_poisson}
    \end{subfigure} 
    \begin{subfigure}[t]{.24\linewidth}
        \centering
    	\includegraphics[width=\textwidth]{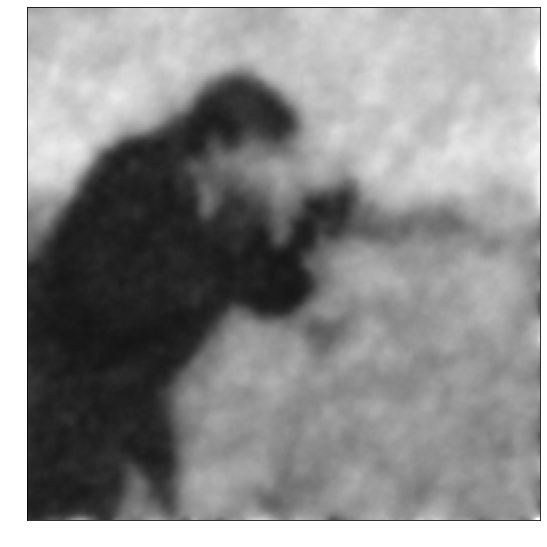}
    	\caption{\small{R-MYULA (20.4 dB)}} 
    \end{subfigure} 
    \begin{subfigure}[t]{.24\linewidth}
        \centering
    	\includegraphics[width=\textwidth]{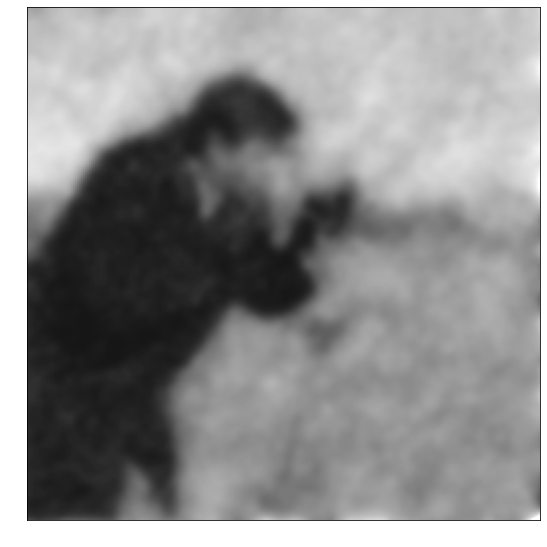}
    	\caption{\small{R-SKROCK (20.5 dB)}} 
    \end{subfigure}
    \begin{subfigure}[t]{.24\linewidth}
        \centering
    	\includegraphics[width=\textwidth]{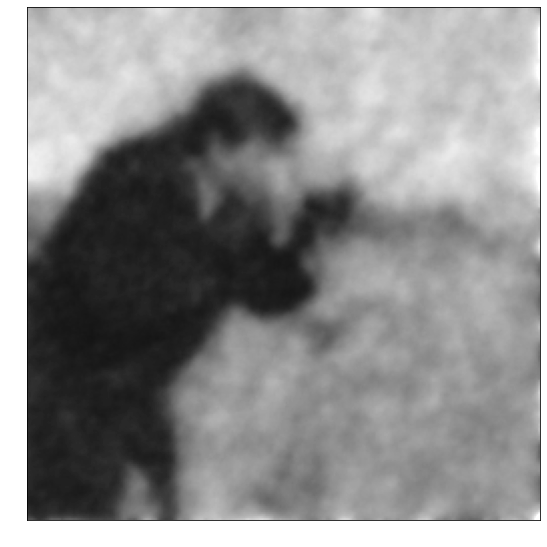}
    	\caption{\small{R-MYUULA (20.4 dB)}} 
    \end{subfigure}
    \begin{subfigure}[t]{.24\linewidth}
        \centering
    	\includegraphics[width=\textwidth]{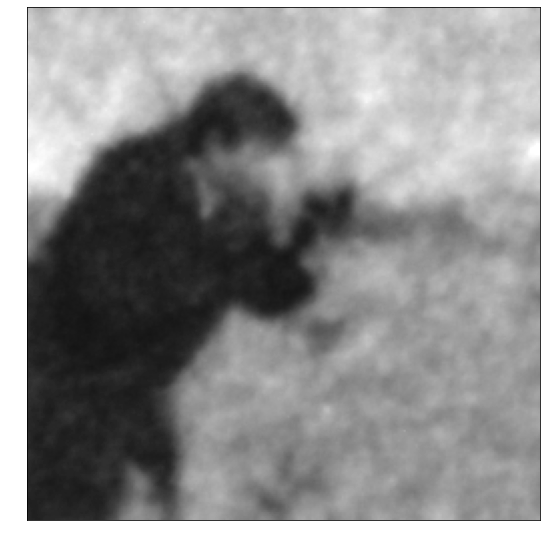}
    	\caption{\small{\hspace{-0.04cm}R-PMALA (20.5 dB)}} 
    \end{subfigure} 
    \begin{subfigure}[t]{.24\linewidth}
        \centering
    	\includegraphics[width=\textwidth]{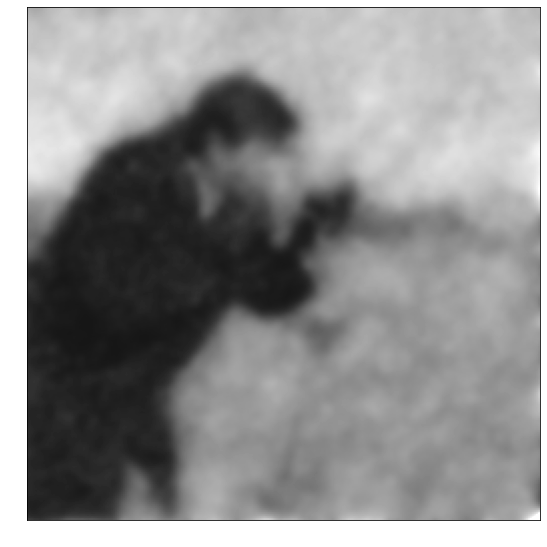}
    	\caption{\small{SPA (20.4 dB)}} 
    \end{subfigure} 
    \begin{subfigure}[t]{.24\linewidth}
        \centering
    	\includegraphics[width=\textwidth]{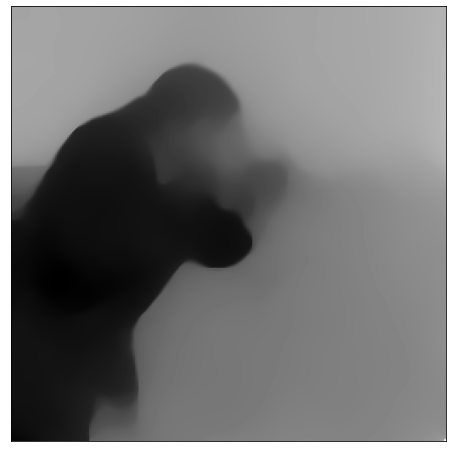}
    	\caption{\small{MAP (19.2 dB)}} 
    \end{subfigure}
    \caption{{\fontfamily{qcr}\selectfont Poisson} experiment: (a) True image $x$ of size $256\times256$; (b) Blurred observation $y$; (c)-(f) MMSE estimators; (h) MAP estimator.}
    \label{fig:post_mean_pois}
\end{figure}

 
\subsection{Non-blind Poisson image deconvolution}
In this experiment, we consider the estimation of a high-resolution image $x\in\mathbb{R}^{n}$ from a blurred observation $y\sim\mathcal{P}(Ax)$, where the blur operator is known. We chose $A$ to be nearly singular leading to an ill-conditioned problem with highly noise-sensitive solutions. The posterior distribution is given by
\begin{equation}\label{posterior_poisson}
\pi \triangleq p(x|y) \propto \exp(-f_{\mathcal{P}}(x)- \theta TV(x))\,\,.
\end{equation}
where we recall that $f_{\mathcal{P}}(\cdot)$ is the Poisson negative log-likelihood defined in \eqref{eq:poissonloglikelihood}.

Figure \ref{fig:true_image_poisson} presents the ground truth {\fontfamily{qcr}\selectfont cameraman} image $x$ of size $n=256\times256$ pixels, whose intensities we scaled for this experiment so that its MIV is $1$ (recall that in the Poisson model the noise power is directly determined by the intensities of the image $x$, so in order to test different signal-to-noise ratios it is necessary to scale the ground truth image). Figure \ref{fig:noisy_image_poisson} shows a realization $y$ with peak signal-to-noise ratio (PSNR) $\sim 5.8$ dB  generated by the observation model \eqref{pois_model} with $A$ being a $5\times5$ uniform blur. For this experiment, the maximum marginal likelihood estimate of $\theta$ is $\hat{\theta} = 5.65$. All the methods are implemented by using this value, and by following the recommendations of Section \ref{guidelines}. For R-PMALA, achieving the optimal acceptance probability of $57\%$ is not possible in this experiment as the maximum admissible value $\delta_{max} = 1/(L_{f^{b}_{y}}+1/\lambda)$ leads to an acceptance probability of $85\%$ (we use this value of $\delta$ for our experiments). Regarding SPA \cite{8683031}, we test a large range of values in a supervised manner and find that $(\rho, \alpha) = (0.035, 0.035)$ provide a competitive trade-off between asymptotic accuracy and convergence speed.

Figure \ref{fig:post_mean_pois} presents also the MMSE estimates given by the proposed methods and the MAP estimate calculated by PIDAL. The MMSE estimates are visually similar with an apparent staircase effect while the MAP estimate returns an oversmoothed reconstruction. It should be noted that this oversmoothed result is related to the particular choice of $\theta$. A value for $\theta$ that improves the MAP estimate can be found (e.g. through cross-validation) but this would require knowledge of the true image (the MMSE restults can also be marginally improved by tailoring $\theta$ in this way). The PSNR values of each method are also reported in Figure \ref{fig:post_mean_pois}. The MMSE estimates have almost identical PSNR values and outperform the MAP estimate in this case.

\begin{figure}[h!]
    \centering
    \begin{subfigure}[t]{0.312\linewidth}
        \centering
    	\includegraphics[width=\textwidth]{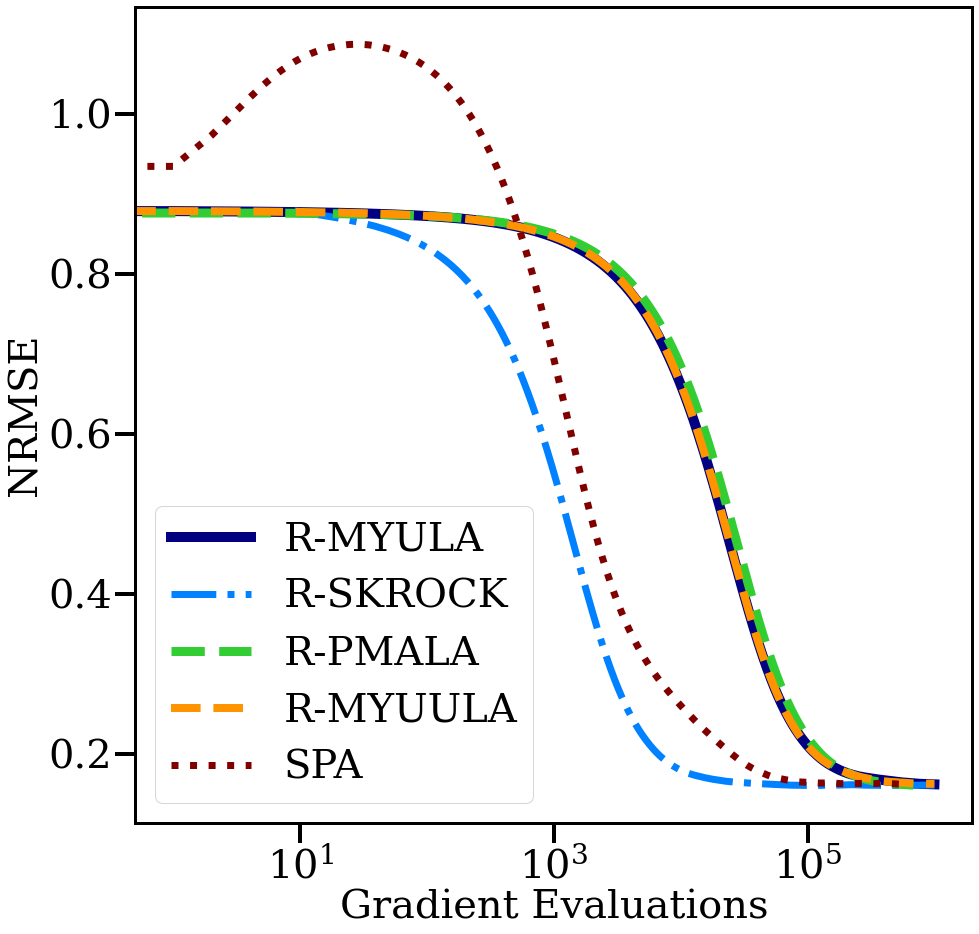}
     	\caption{}
     	\label{fig:NRMSE_pois}
    \end{subfigure} 
    \begin{subfigure}[t]{0.32\linewidth}
        \centering
    	\includegraphics[width=\textwidth]{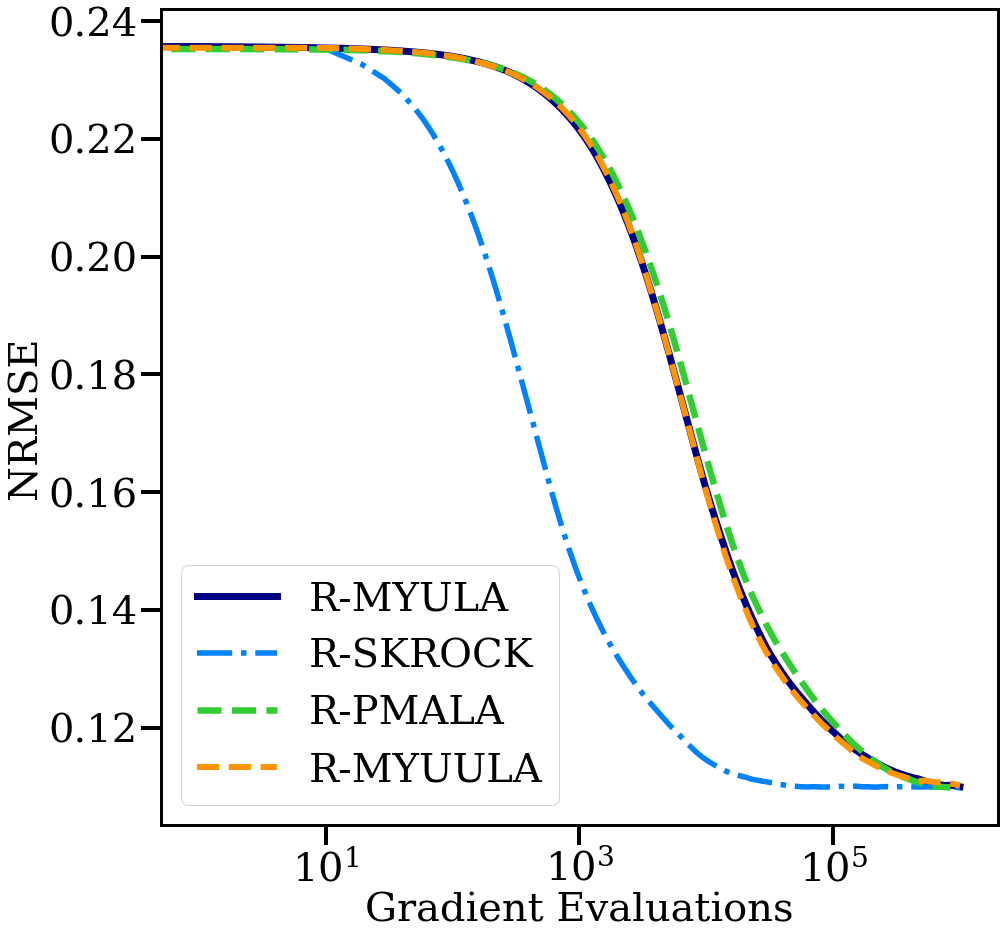}
    	\caption{}
    	\label{fig:NRMSE_bin}
    \end{subfigure} 
    \begin{subfigure}[t]{0.32\linewidth}
        \centering
    	\includegraphics[width=\textwidth]{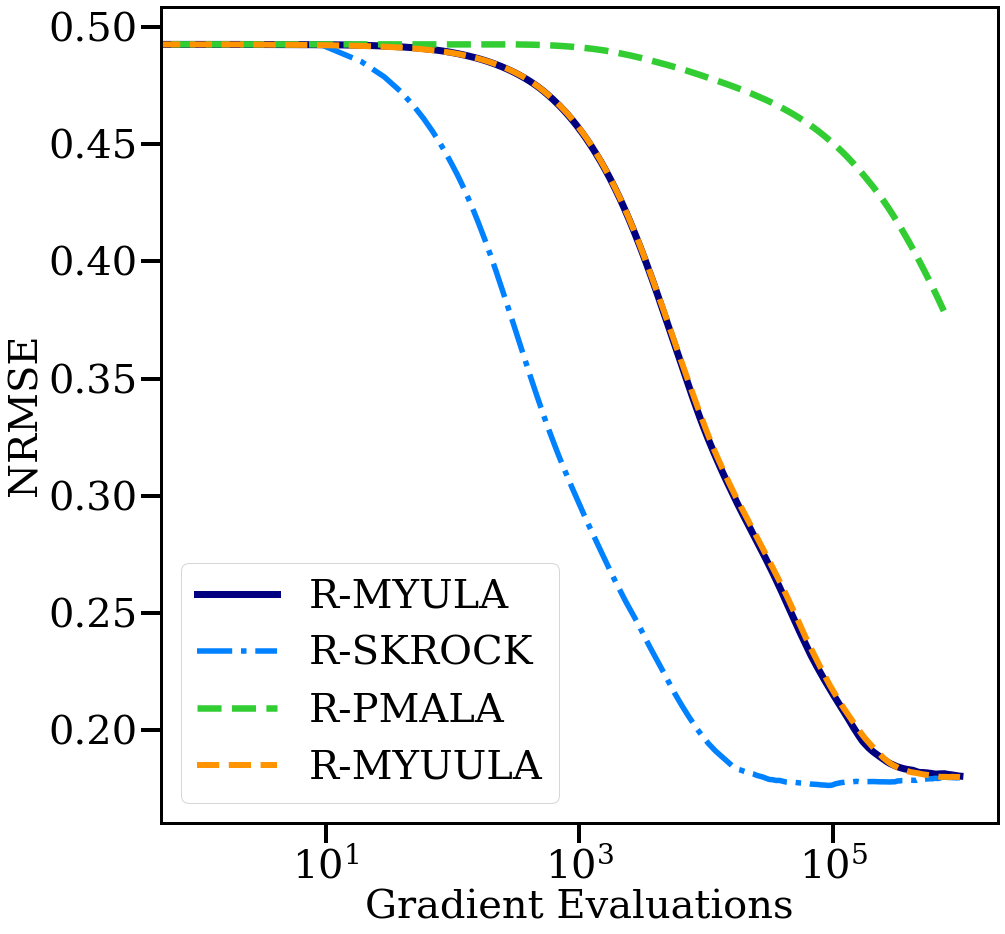}
    	\caption{}
    	\label{fig:NRMSE_geo}
    \end{subfigure} 
    \vspace{-0.2cm}
    \caption{Evaluation of NRMSE for (a) the Poisson deblurring experiment (b) the binomial denoising experiment and (c) the geometric inpainting experiment. The number of gradient evaluations are presented in log-scale (base 10).}
    \label{fig:NRMSEs}
\end{figure}

Figure \ref{fig:NRMSE_pois} presents the evolution of the NRMSE estimation for the MMSE solutions as a function of the number of gradient evaluations (in log-scale). We observe that R-SKROCK has the highest convergence speed followed by SPA whose NRMSE increases for a small number of iterations. This is due to a transient regime of the auxiliary variables introduced in SPA. It can also be observed that R-MYULA and R-MYUULA have similar convergence speed, and R-PMALA has slightly slower convergence speed due to its Metropolised nature.

\begin{figure}[t!]
    \centering
    \begin{subfigure}[t]{\linewidth}
        \centering
    	\includegraphics[width=\textwidth]{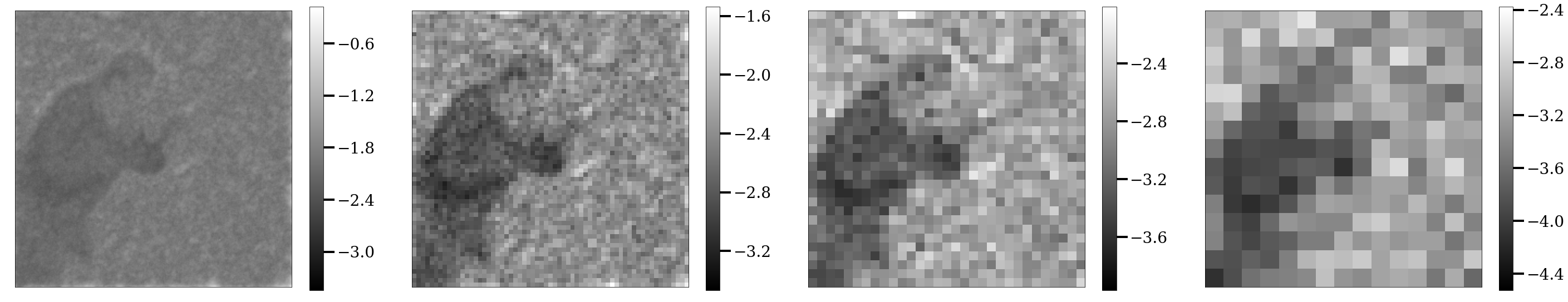}
     	\caption{R-MYULA: standard deviation (log-scaled) at different scales.}
    \end{subfigure}
    \begin{subfigure}[t]{\linewidth}
        \centering
    	\includegraphics[width=\textwidth]{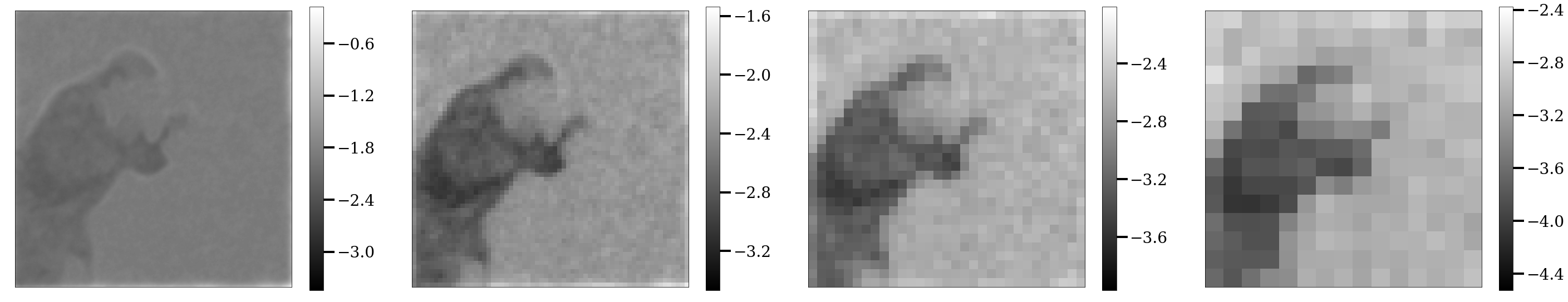}
    	\caption{R-SKROCK: standard deviation (log-scaled) at different scales.}
    \end{subfigure} 
    \begin{subfigure}[t]{\linewidth}
        \centering
    	\includegraphics[width=\textwidth]{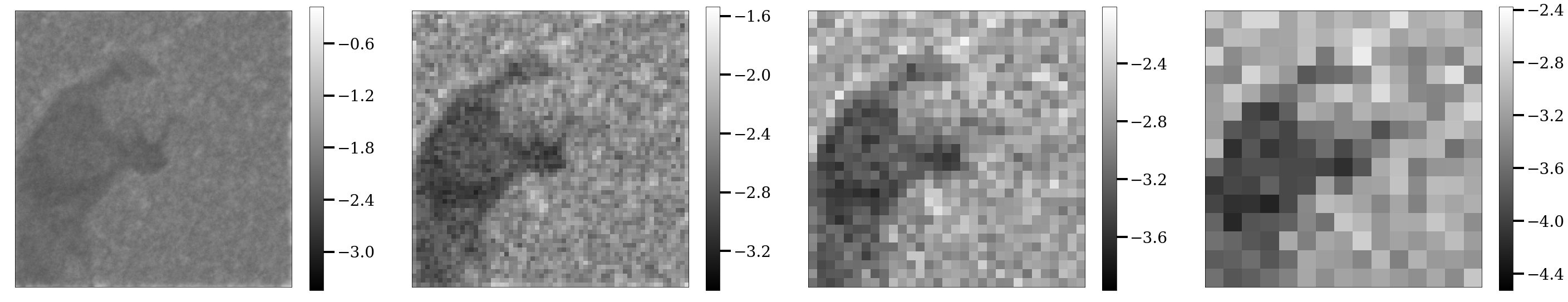}
    	\caption{R-MYUULA: standard deviation (log-scaled) at different scales.}
    \end{subfigure} 
    \begin{subfigure}[t]{\linewidth}
        \centering
    	\includegraphics[width=\textwidth]{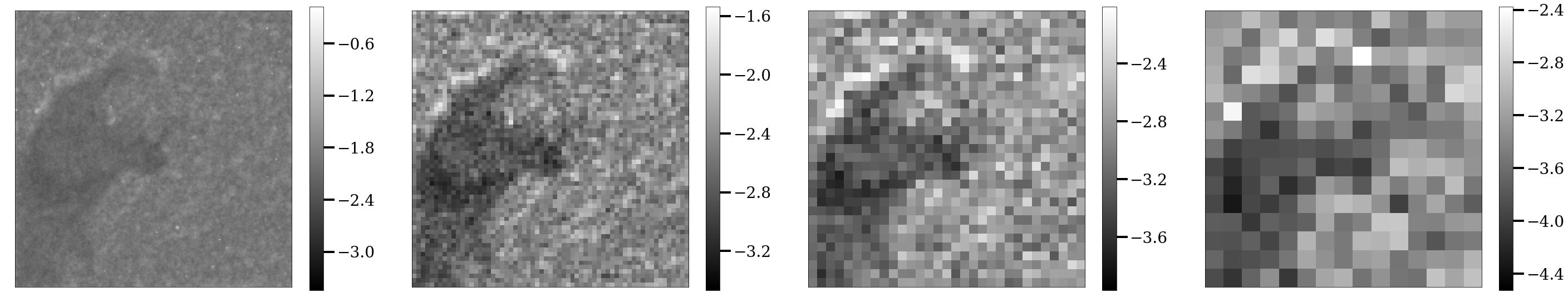}
    	\caption{R-PMALA: standard deviation (log-scaled) at different scales.}
    \end{subfigure} 
    \caption{Marginal posterior standard deviation for the Poisson deblurring problem at different scales (0,2,4,8 from left to right). The scale $i$ corresponds to a downsampling by a factor $2i$ of the original sample size. Most of the uncertainty is located around the edges and the cameraman's background.}
    \label{fig:st_dev_poisson}
\end{figure}

In Figure \ref{fig:st_dev_poisson}, we present uncertainty visualization plots that are useful to quantify the uncertainty related to image structures at different spatial scales. Specifically,  for different scales $j$, we downsample (by averaging) the stored samples by a factor of $2j$ before computing the standard deviation. The estimates of the pixel-wise standard deviations ($j=0$) obtained by each algorithm are presented in the first column of Figure \ref{fig:st_dev_poisson}. It is observed that the estimates obtained by R-MYULA, R-PMALA and R-MYUULA are less accurate of the respective one obtained by R-SKROCK in agreement with the experiments under Gaussian noise in \cite{vargas2020accelerating}. High uncertainty is concentrated around the edges of the cameraman which is expected from the particular choice of prior. Some uncertainty also exists within the background of the scene, since the likelihood is highly uninformative for these pixel regions and the prior does not carry any information to recover important details (e.g. the buildings). Note that because of the Poisson likelihood and the positivity constraints, the uncertainties appear larger in the background (higher average intensity) than on the coat of the cameraman (lower average intensity). However, the relative uncertainties (i.e., when dividing by the pixelwise true intensities) would be higher in dark than bright pixels. For larger structures of pixels ($j \neq 0$), the highest uncertainty is observed at the background of the cameraman's image, since there both the likelihood and the prior are the most weak.

\begin{figure}[h!]
    \centering
    \begin{subfigure}[t]{0.32\linewidth}
        \centering
    	\includegraphics[width=\textwidth]{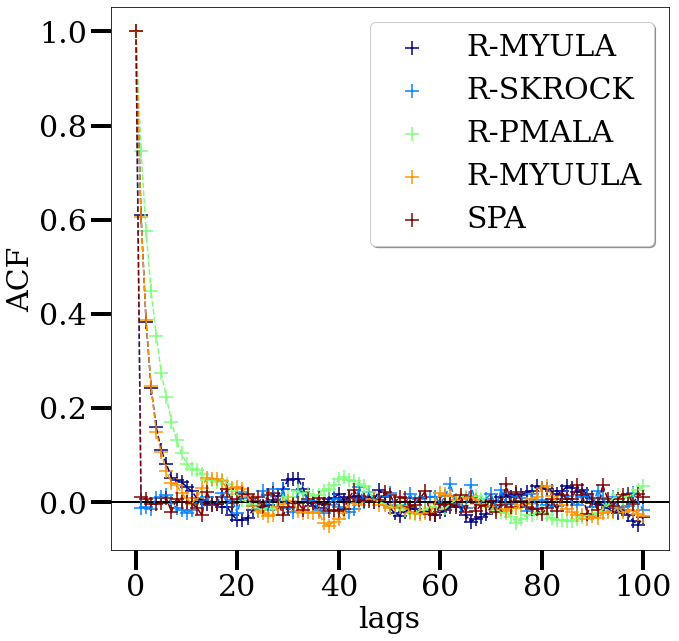}
     	\caption{Fastest direction}
    \end{subfigure} 
    \begin{subfigure}[t]{0.32\linewidth}
        \centering
    	\includegraphics[width=\textwidth]{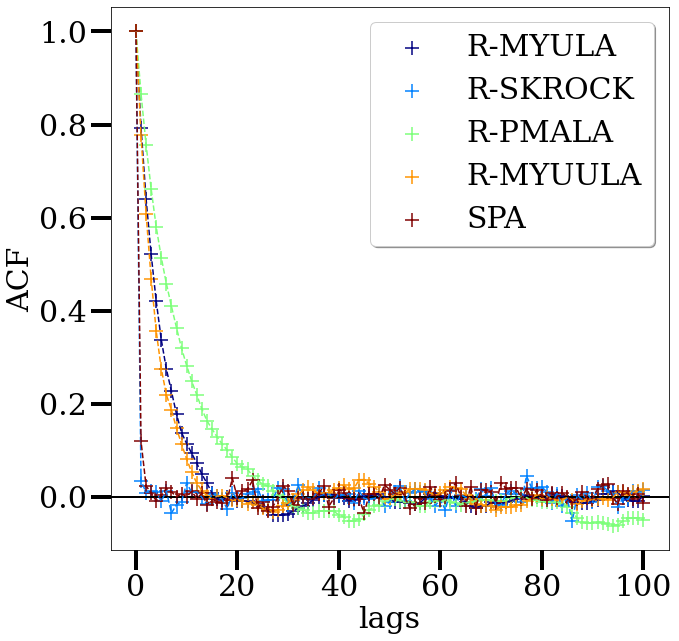}
    	\caption{Median direction}
    \end{subfigure} 
    \begin{subfigure}[t]{0.32\linewidth}
        \centering
    	\includegraphics[width=\textwidth]{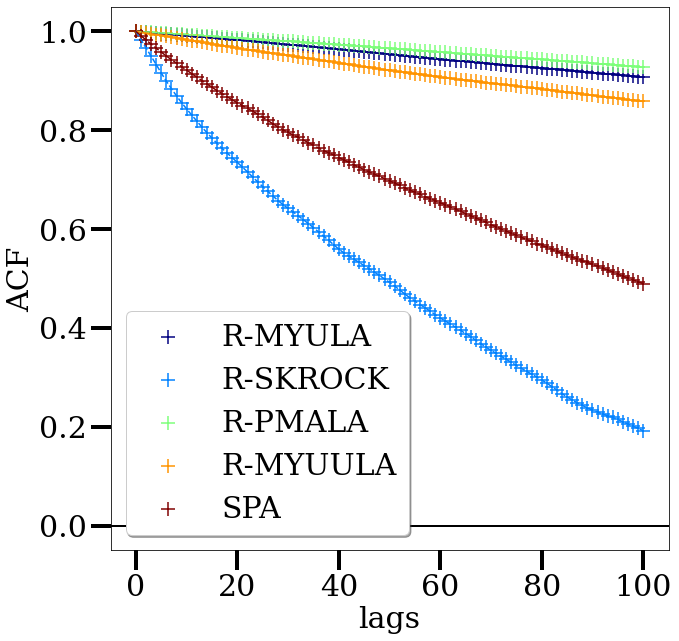}
    	\caption{Slowest direction}
    \end{subfigure} 
    \caption{ACF for the fastest, median, and slowest direction for the Poisson deblurring problem. The ACF is shown for lags up to 100 for all images in the Fourier domain.}
    \label{fig:autocor_pois}
\end{figure}

\begin{table}[h!]
    \footnotesize
 	\centering
    \scalebox{1}{\begin{tabular}{| c | c | c | c | c |}
    \hline
    \textbf{Poisson deblurring cases} & \multicolumn{2}{c |}{\bfseries MMSE} &  \multicolumn{2}{c |}{\bfseries MAP} \\ 
    \cline{2-5}
    & \textbf{NRMSE} & \textbf{PSNR} & \textbf{NRMSE} & \textbf{PSNR} \\
    \hline
    \hspace*{-0.5cm} $\text{MIV} = 1$, $\text{g.e.}=1\times10^{6}$ &  \textbf{0.1612} &  \textbf{20.53 dB} &  0.1873 & 19.25 dB  \\ \hline
    \hspace*{-0.35cm} $\text{MIV} = 10$, $\text{g.e.}=2\times10^{6}$ &  \textbf{0.1143} & \textbf{23.54 dB}  & 0.1186  & 23.22 dB \\ \hline
    $\text{MIV} = 100$, $\text{g.e.}=5\times10^{6}$ & \textbf{0.088}  & \textbf{25.74 dB} & 0.089 & 25.72 dB \\ \hline
    \end{tabular}}
    \caption{Comparison of MMSE and MAP estimator for different Poisson deblurring cases. After a large number of gradient evaluations (g.e.), the algorithms exhibit identical performance.} 
    \label{table_pois}
\end{table}

To conclude this experiment, the sample ACFs are computed for the slowest and the fastest directions in the Fourier domain\footnote{The slowest (fastest) direction corresponds to the Fourier coefficient with the highest (lowest) variance.}. In Figure \ref{fig:autocor_pois}, we see that independence is reached in approximately 20 iterations (after thinning) in the median and fastest directions, and it is much slower for the few very uncertain coefficients. In addition, the superiority of R-SKROCK and SPA in sense of convergence properties can be observed along all directions. The ACF of the R-SKROCK samples along the slowest direction decay the fastest among all the methods.

Finally, Table \ref{table_pois} reports comparisons between the MMSE estimation and the MAP estimation under less severe Poisson noise (i.e. higher MIV for the target image). We run the MCMC algorithms enough time so the calculated MMSE estimates of each method have identical performance. It is observed that the MMSE estimate outperforms the MAP estimate, particularly in cases where the likelihood is uninformative (i.e. low MIV).

\begin{figure}[h!]
    \centering
    \begin{subfigure}[t]{.24\linewidth}
        \centering
    	\includegraphics[width=\textwidth]{images/ground_truth.png}
    	\caption{\small{True image $x$}}
    \end{subfigure} 
    \begin{subfigure}[t]{.24\linewidth}
        \centering
    	\includegraphics[width=\textwidth]{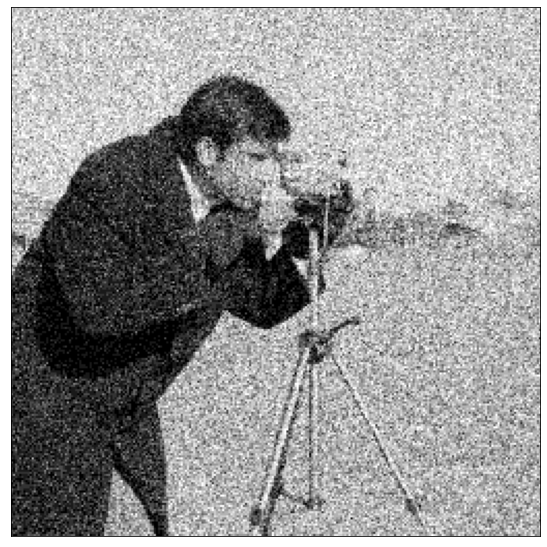}
    	\caption{\small{Observation $y$ (17.25 dB)}}
    	\label{fig:bin_noise_image}
    \end{subfigure} 
    \begin{subfigure}[t]{.24\linewidth}
        \centering
    	\includegraphics[width=\textwidth]{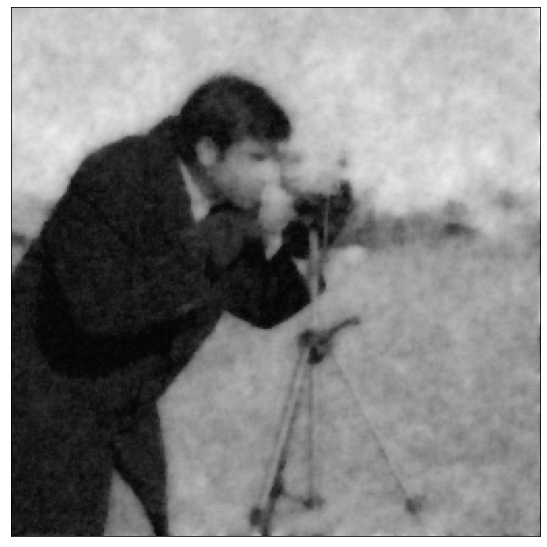}
    	\caption{\small{R-MYULA (23.8 dB)}} 
    \end{subfigure} 
    \begin{subfigure}[t]{.24\linewidth}
        \centering
    	\includegraphics[width=\textwidth]{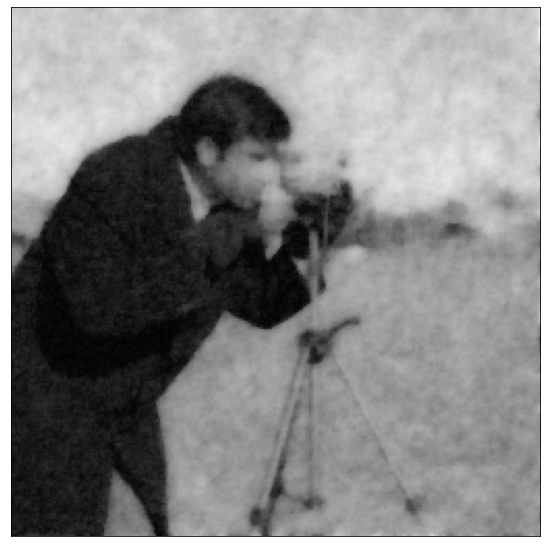}
    	\caption{\small{R-SKROCK (23.9 dB)}} 
    \end{subfigure}
    \begin{subfigure}[t]{.24\linewidth}
        \centering
    	\includegraphics[width=\textwidth]{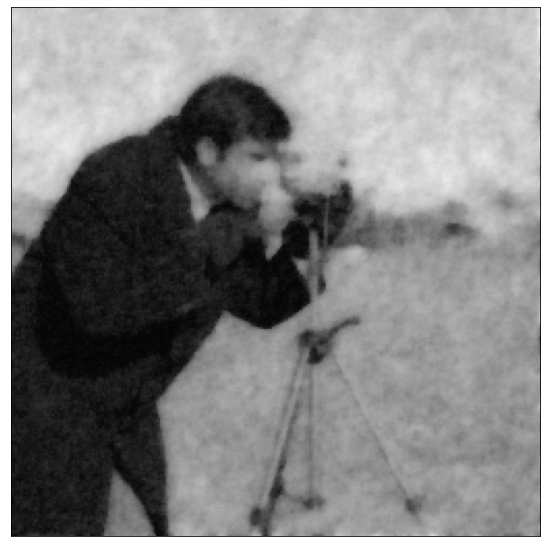}
    	\caption{\small{R-MYUULA (23.8 dB)}} 
    \end{subfigure}
    \begin{subfigure}[t]{.24\linewidth}
        \centering
    	\includegraphics[width=\textwidth]{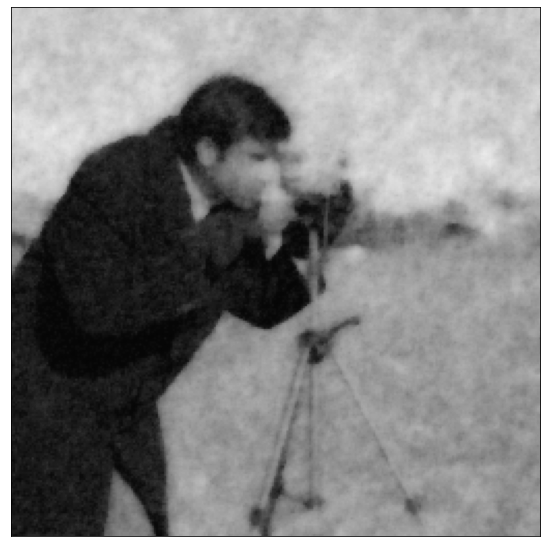}
    	\caption{\small{\hspace{-0.04cm}R-PMALA (23.8 dB)}} 
    \end{subfigure} 
    \begin{subfigure}[t]{.24\linewidth}
        \centering
    	\includegraphics[width=\textwidth]{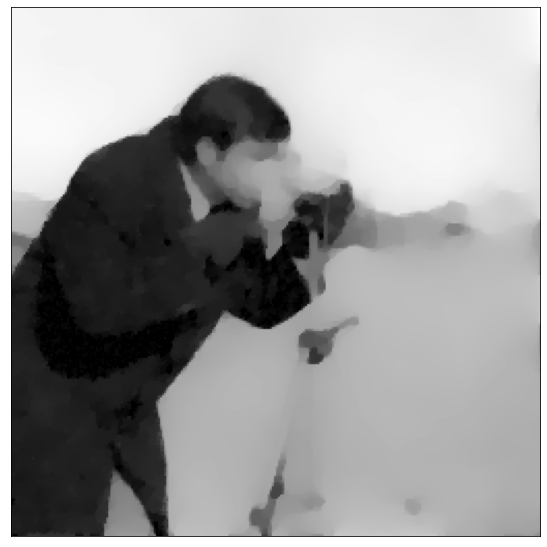}
    	\caption{\small{MAP (23.2 dB)}}
    \end{subfigure}
    \caption{{\fontfamily{qcr}\selectfont Binomial} experiment: (a) True image $x$; (b) Observation $y$; (c)-(f)  MMSE estimators; (g) MAP estimator.}
    \label{fig:post_mean_bin}
\end{figure}

\subsection{Binomial denoising}

We consider now a binomial denoising problem, where we seek to estimate a high-resolution image $x\in\mathbb{R}^{n}$ from a realization $y\sim\mathcal{B}in(t,1-e^{-x})$. The posterior distribution is given by
\begin{equation}\label{posterior_bin}
\pi \triangleq p(x|y) \propto \exp(-f_{\mathcal{B}}(x)- \theta TV(x))~~,
\end{equation}
where $f_{\mathcal{B}}(\cdot)$ is the binomial negative log-likelihood in \eqref{eq:bin_likelihood}.


\noindent Figure \ref{fig:bin_noise_image} presents a binomial realization $y$ of the {\fontfamily{qcr}\selectfont cameraman} image generated by using the observation model \eqref{bin_model} with $A=I$, $t_{i} = 10 \hspace{0.2cm}\forall i=1,\ldots,n$ and the MIV of $x$ equal to $1$ (recall that in the binomial model the noise power is directly determined by the intensities of the image x and the repetition periods $t$). For this experiment, the maximum marginal likelihood estimate of $\theta$ is $\hat{\theta}=6.4$.  All the methods are implemented by using this value, and by following the recommendations of Section \ref{guidelines}. For R-PMALA, we achieved an acceptance rate of $83\%$ by setting $\delta=\delta_{max}=1/(L_{y}^{b}+1/\lambda)$.

\begin{figure}[b!]
    \centering
    \begin{subfigure}[t]{\linewidth}
        \centering
    	\includegraphics[width=\textwidth]{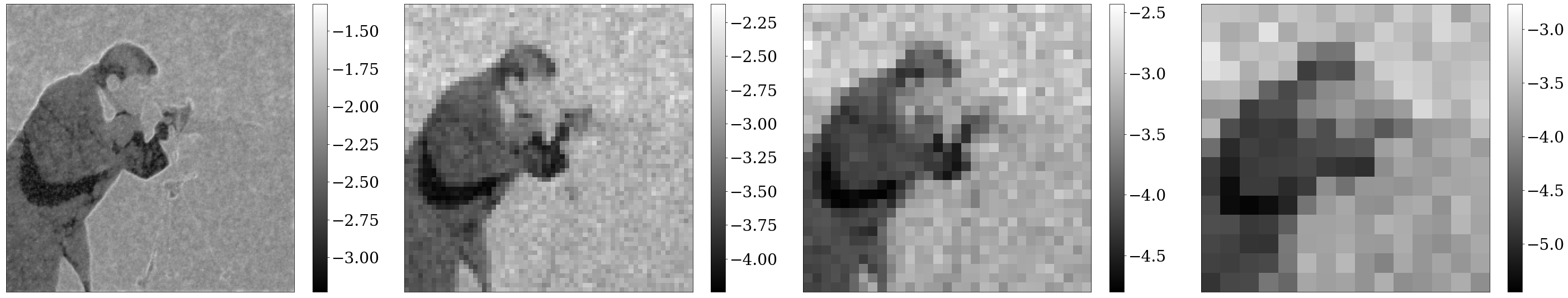}
     	\caption{R-MYULA: standard deviation (log-scaled)  at different scales.}
    \end{subfigure}
    \begin{subfigure}[t]{\linewidth}
        \centering
    	\includegraphics[width=\textwidth]{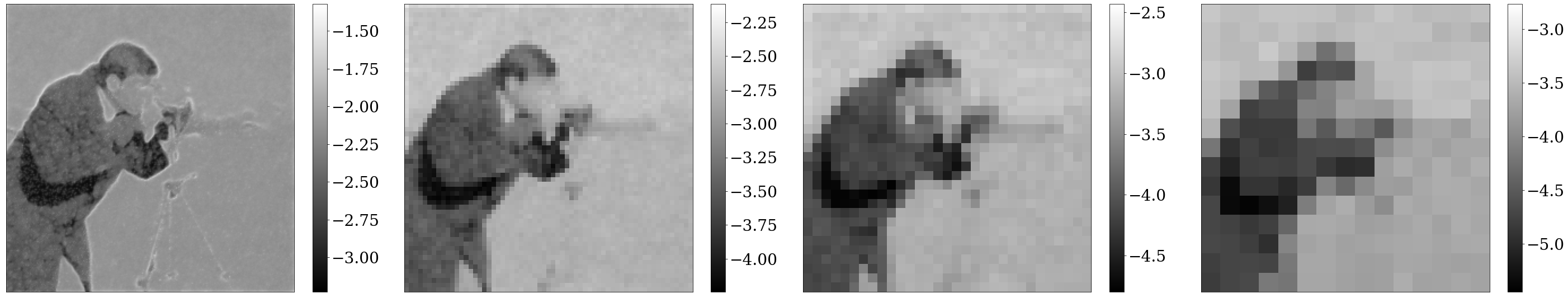}
    	\caption{R-SKROCK: standard deviation (log-scaled) at different scales.}
    \end{subfigure} 
    \begin{subfigure}[t]{\linewidth}
        \centering
    	\includegraphics[width=\textwidth]{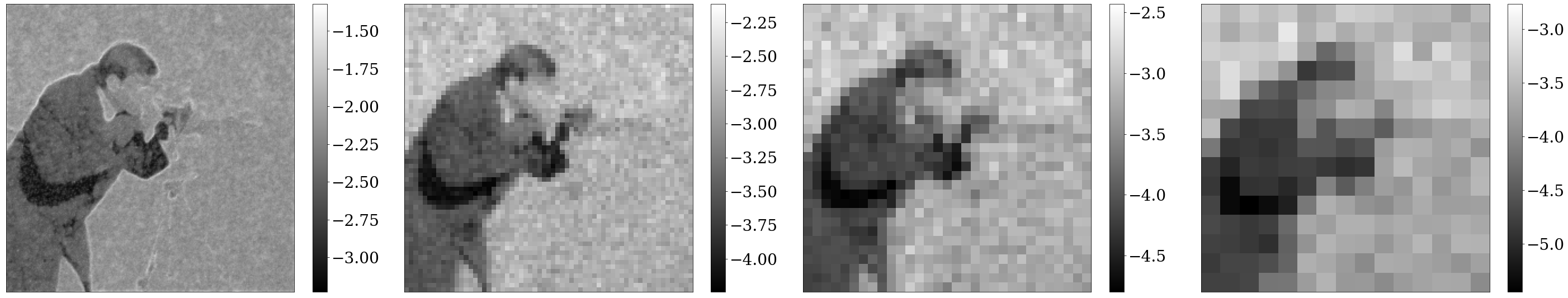}
    	\caption{R-MYUULA: standard deviation (log-scaled) at different scales.}
    \end{subfigure} 
    \begin{subfigure}[t]{\linewidth}
        \centering
    	\includegraphics[width=\textwidth]{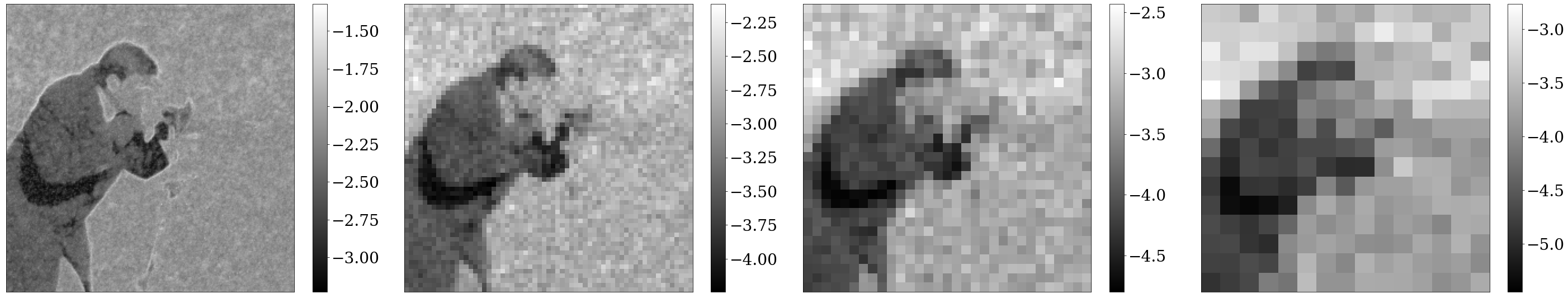}
    	\caption{R-PMALA: standard deviation (log-scaled) at different scales.}
    \end{subfigure} 
    \caption{Marginal posterior standard deviation for the binomial denoising problem at different scales (0,2,4,8 from left to right). The scale $i$ corresponds to a downsampling by a factor $2i$ of the original sample size. Most of the uncertainty is located around the edges and the cameraman's background.}
    \label{fig:st_dev_bin}
\end{figure}

Figure \ref{fig:post_mean_bin} also presents the MMSE estimates given by the proposed methods and the MAP estimate calculated by adjusting PIDAL for binomial likelihoods. The MMSE estimates are visually similar with an apparent staircase effect while the MAP estimate returns an oversmoothed reconstruction. The PSNR values of each method are reported in the respective captions in Figure \ref{fig:post_mean_bin}. The MMSE estimates have almost identical PSNR values which are higher than the respective one of the MAP estimate. 


Figure \ref{fig:NRMSE_bin} presents the evolution of the NRMSE estimation for the MMSE solutions as a function of the number of gradient evaluations (in log-scale). We observe that R-SKROCK has the highest convergence speed, R-MYULA and R-MYUULA have similar convergence speed, and R-PMALA has the slowest convergence among the proposed methods.

The estimates of the pixel-wise standard deviations ($j=0$) are presented in the first column of Figure \ref{fig:st_dev_bin}. It is observed that the estimates obtained by R-MYULA, R-PMALA and R-MYUULA are slightly more noisy than respective one obtained by R-SKROCK. High uncertainty is concentrated around the edges (e.g. cameraman figure and buildings) and some uncertainty also exists in homogeneous regions, since the likelihood is quite uninformative for these pixel regions. For larger structures of pixels ($j \neq 0$), the highest uncertainty is spotted at the background part of the cameraman's image. see Figure \ref{fig:st_dev_bin}. R-SKROCK returns a smooth estimate of the posterior variance while R-PMALA returns a slightly noiser estimate since its Metropolised nature leads to higher estimation variance \cite{2018,durmus2016efficient}.

\begin{figure}[h!]
    \centering
    \begin{subfigure}[t]{0.32\linewidth}
        \centering
    	\includegraphics[width=\textwidth]{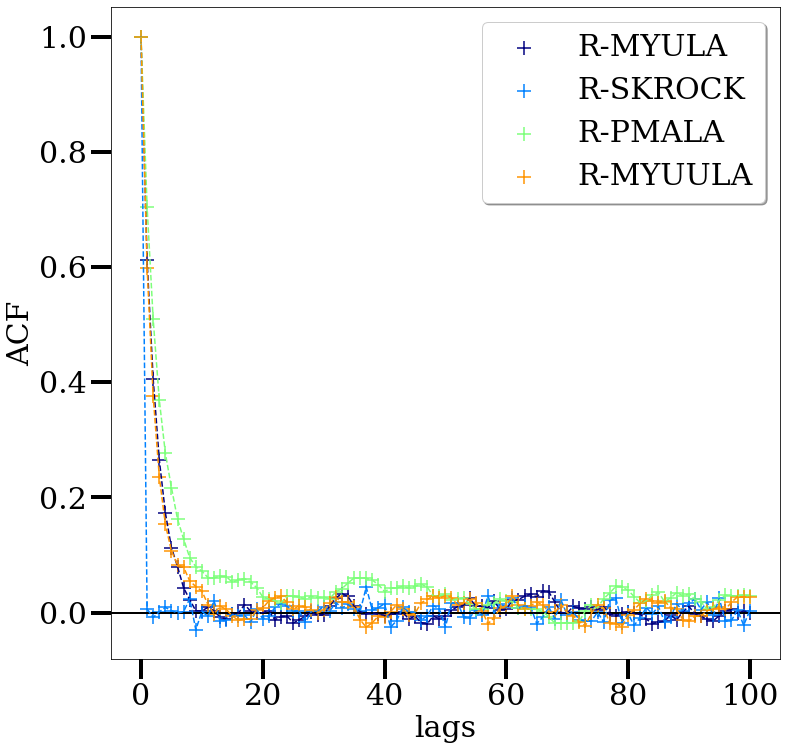}
     	\caption{Fastest component}
    \end{subfigure} 
    \begin{subfigure}[t]{0.32\linewidth}
        \centering
    	\includegraphics[width=\textwidth]{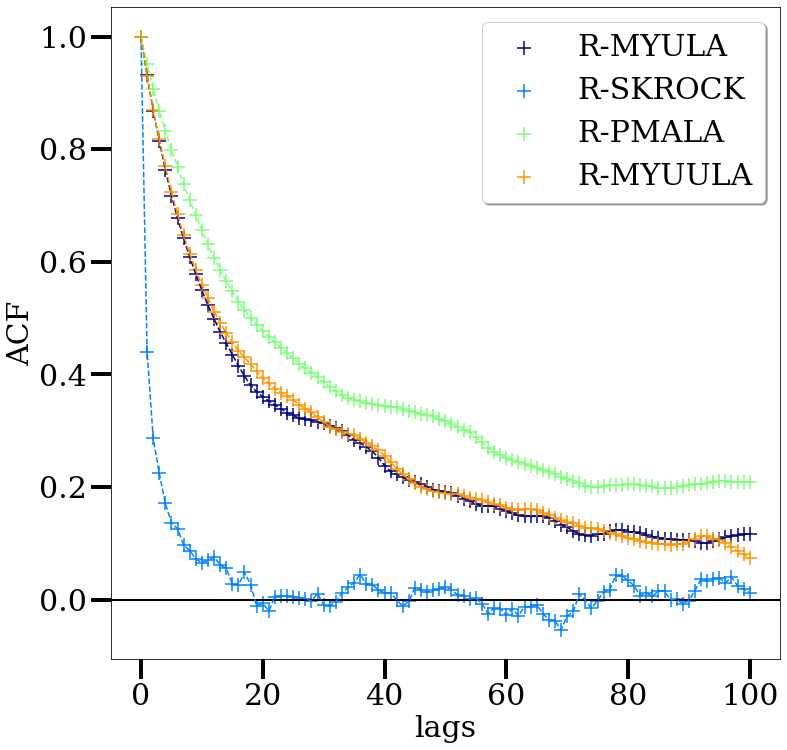}
    	\caption{Median component}
    \end{subfigure} 
    \begin{subfigure}[t]{0.32\linewidth}
        \centering
    	\includegraphics[width=\textwidth]{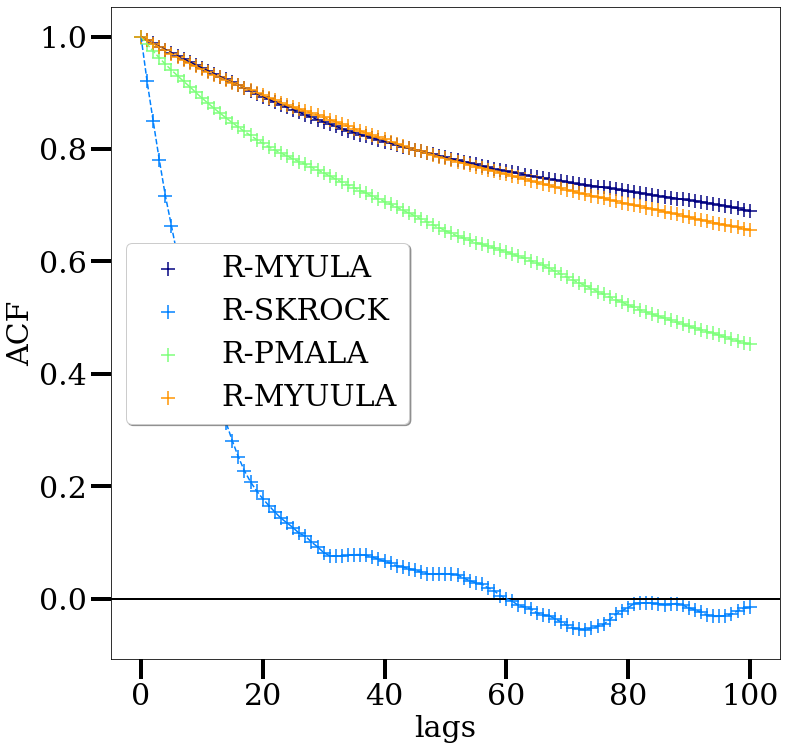}
    	\caption{Slowest component}
    \end{subfigure} 
    \caption{ACF for the fastest, median, and slowest components for the binomial denoising problem. The ACF is shown for lags up to 100 for all images in the pixel domain.}
    \label{fig:autocor_bin}
\end{figure}

Finally, in Figure \ref{fig:autocor_bin}, we see that independence is reached fast for the components of low or median uncertainty, and is much slower for the few very uncertain pixels. The superiority of R-SKROCK in sense of convergence properties compared with the other methods can also be observed. Particularly, the ACF of the R-SKROCK samples along the slowest component decay the fastest among the ones given by R-MYULA, R-MYUULA, and R-PMALA.

Table \ref{table_bin} reports NRMSE and PSNR comparisons between the MMSE estimation and the MAP estimation for different cases of binomial noise. We run the MCMC algorithms enough time so the calculated MMSE estimates of each method have identical performance. It is observed that the MMSE estimator outperforms the MAP estimator in cases where the likelihood is quite uninformative (low MIV or low number of repetition periods $t$) while the MAP estimator is more competitive under more informative likelihoods.

\begin{table}[h!]
    \footnotesize
 	\centering
    \scalebox{1}{\begin{tabular}{| c | c | c | c | c |}
    \hline
    \textbf{Binomial denoising cases} & \multicolumn{2}{c |}{\bfseries MMSE} &  \multicolumn{2}{c |}{\bfseries MAP} \\ 
    \cline{2-5}
    & \textbf{NRMSE} & \textbf{PSNR} & \textbf{NRMSE} & \textbf{PSNR} \\
    \hline
    MIV = 1 , t=10 , g.e.$=5\times 10^{6}$ &  \textbf{0.1092} &  \textbf{23.94 dB} &  0.1189 & 23.20 dB  \\ \hline
    \hspace{0.07cm} MIV = 1 , t=100 , g.e.$=5\times 10^{6}$  &  0.0592 &  29.25 dB &  \textbf{0.0573} & \textbf{29.54 dB}  \\ \hline
    \hspace{0.37cm} MIV = 0.1 , t=100 , g.e.$=3\times 10^{6}$ &  \textbf{0.0946} & \textbf{25.19 dB}  &  0.099 &  24.76 dB\\ \hline
    \hspace{0.58cm} MIV = 0.1 , t=1000 , g.e.$=5\times 10^{6}$ & 0.0496  & 30.79 dB & \textbf{0.045} & \textbf{31.66 dB} \\ \hline
    \end{tabular}}
    \caption{Comparison of MMSE and MAP estimator for different binomial denoising cases. After a large number of gradient evaluations (g.e.), the algorithms exhibit identical performance.}
    \label{table_bin}
    \vspace{-0.4cm}
\end{table}

\subsection{Geometric Inpainting}

We consider now a very challenging inpainting problem, where we seek to estimate a high-resolution image $x\in\mathbb{R}^{n}$ from a set of geometric  measurements $t\sim\mathcal{G}eo(1-e^{-Ax})$, where $A$ is a $m\times n$ matrix containing $m\leq n$ randomly selected rows of the $n\times n$ identity matrix. The posterior distribution is given by
\begin{equation}\label{posterior_geo}
\pi \triangleq p(x|t) \propto \exp(-f_{\mathcal{G}}(x)- \theta TV(x))~~,
\end{equation}
where $f_{\mathcal{G}}(\cdot)$ is the geometric log-likelihood in \eqref{geo_model}.

\begin{figure}[h!]
    \centering
    \begin{subfigure}[t]{.24\linewidth}
        \centering
    	\includegraphics[width=\textwidth]{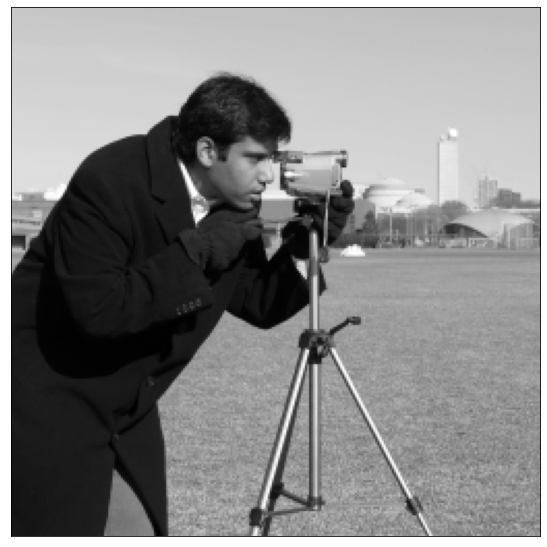}
    	\caption{\small{True image $x$}}
    \end{subfigure} 
    \begin{subfigure}[t]{.24\linewidth}
        \centering
    	\includegraphics[width=\textwidth]{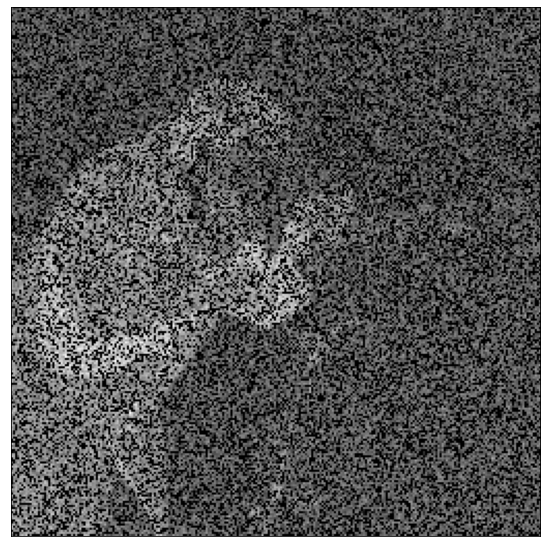}
    	\caption{\small{Observation $\log(A^{T}t+1)$}}
    \end{subfigure} 
    \begin{subfigure}[t]{.24\linewidth}
        \centering
    	\includegraphics[width=\textwidth]{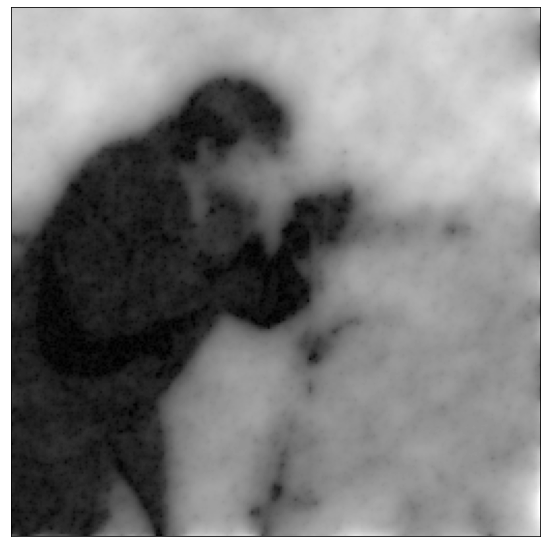}
    	\caption{\small{R-MYULA (19.6 dB)}} 
    \end{subfigure} 
    \begin{subfigure}[t]{.24\linewidth}
        \centering
    	\includegraphics[width=\textwidth]{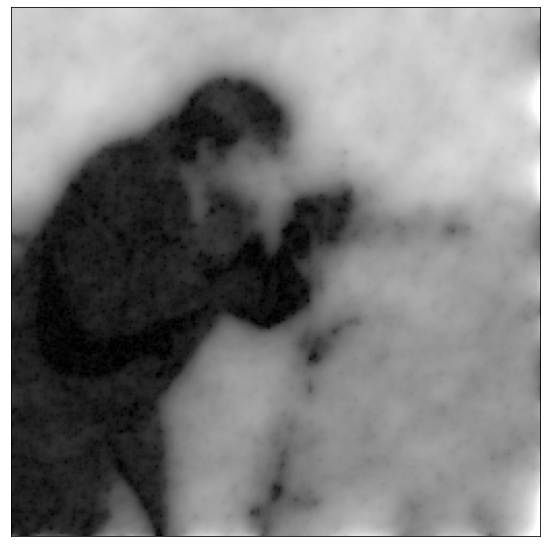}
    	\caption{\small{R-SKROCK (19.6 dB)}} 
    \end{subfigure} 
    \begin{subfigure}[t]{.24\linewidth}
        \centering
    	\includegraphics[width=\textwidth]{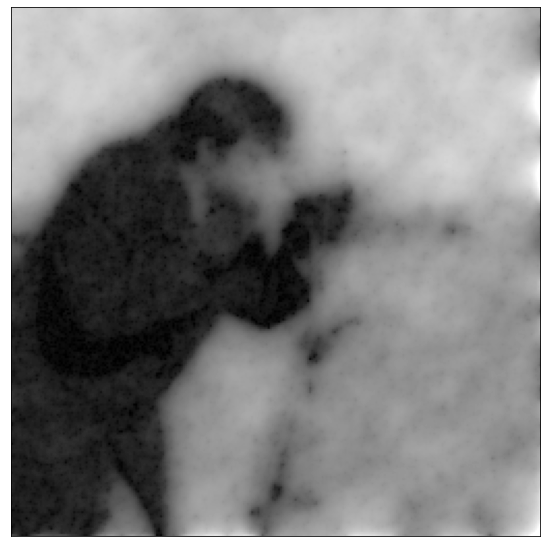}
    	\caption{\small{R-MYUULA (19.6 dB)}} 
    \end{subfigure}
    \begin{subfigure}[t]{.24\linewidth}
        \centering
    	\includegraphics[width=\textwidth]{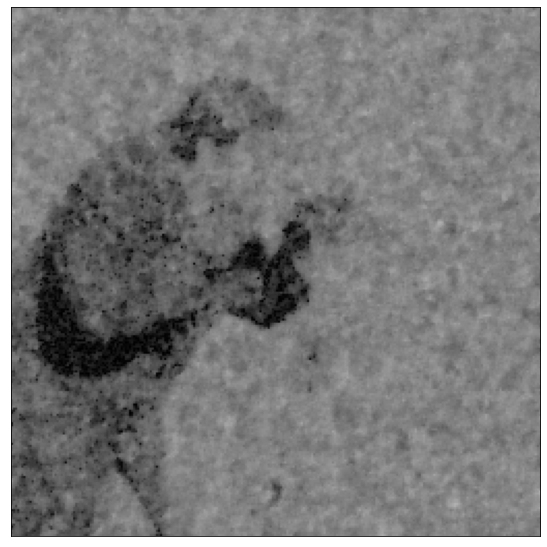}
    	\caption{\small{\hspace{-0.05cm}R-PMALA (13.3 dB)}} 
    \end{subfigure} 
    \begin{subfigure}[t]{.238\linewidth}
        \centering
    	\includegraphics[width=\textwidth]{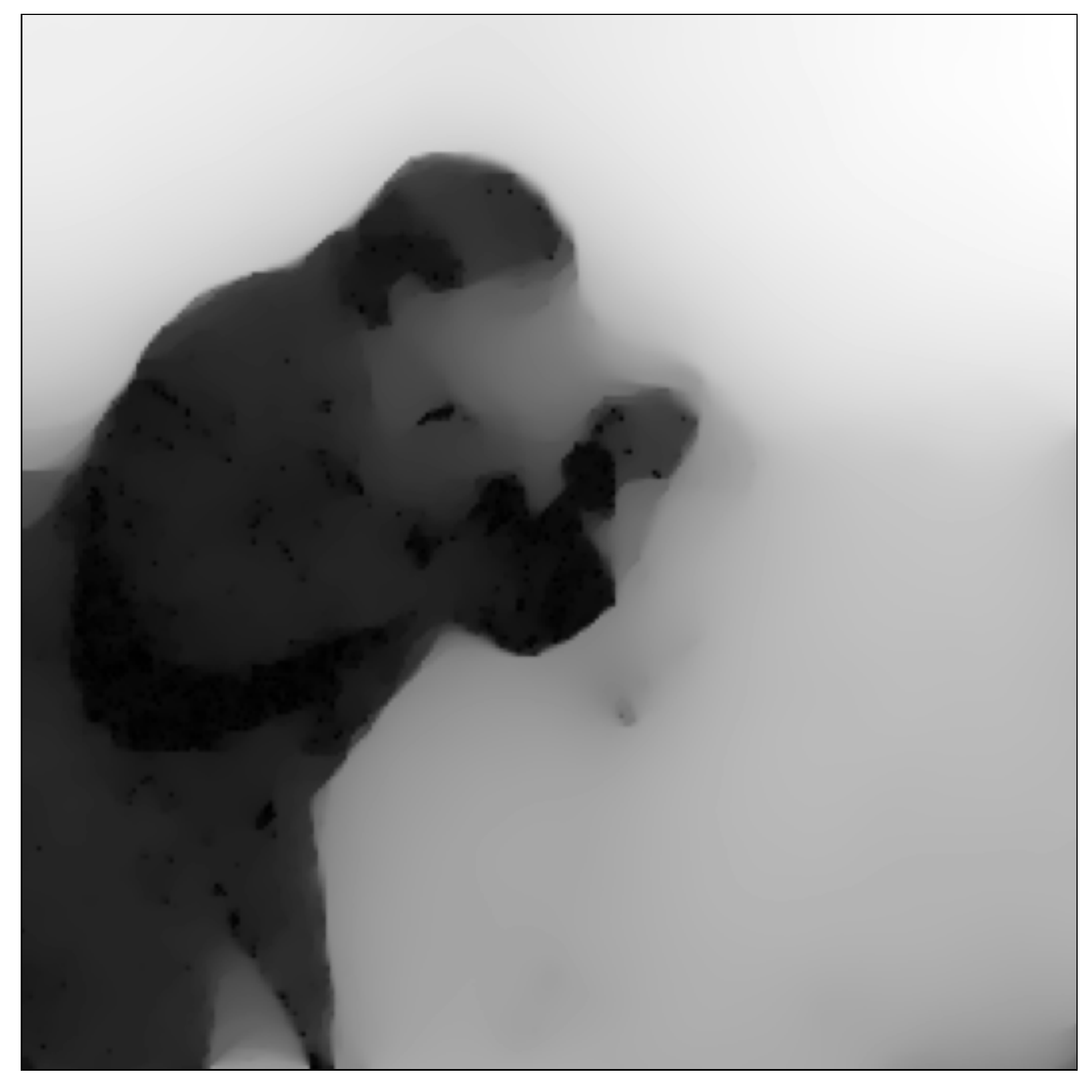}
    	\caption{\small{MAP (17.6 dB)}} 
    \end{subfigure}
    \caption{{\fontfamily{qcr}\selectfont Geometric} experiment: (a) True image $x$; (b) Observation $A^{T}t$ is displayed as  $\log(A^{T}t + 1)$ for visualization purposes; (c)-(f)  MMSE estimators as computed by the different MCMC algorithms; (g) MAP estimator.}
    \label{fig:post_mean_geo}
\end{figure}

Figure \ref{fig:post_mean_geo} presents a (backprojected) realization of $x$ generated by using the model \eqref{geo_model} with $A$ being a masking operator with $30\%$ mixing pixels and the MIV of $x$ being $10^{-2}$. For this experiment, the maximum marginal likelihood estimate of $\theta$ is $\hat{\theta}=632$. All the methods are implemented by using this value, and by following the recommendations of Section \ref{guidelines}. Regarding R-PMALA, it struggled to reach stationarity up to the available computational time because of the extremely uninformative likelihood. It required to take a very small step size ($\delta<<\delta_{max}$) to reach an acceptable acceptance rate of $57\%$.

\begin{figure}[h!]
    \centering
    \begin{subfigure}[t]{\linewidth}
        \centering
    	\includegraphics[width=\textwidth]{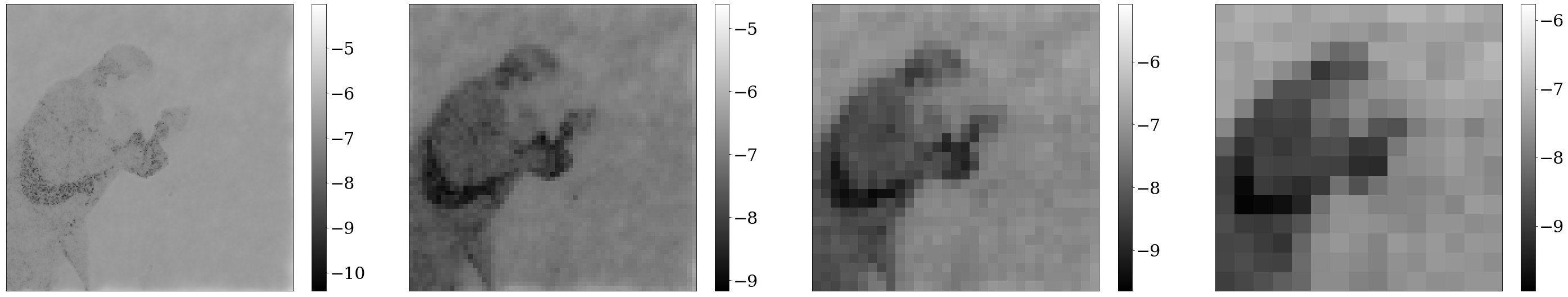}
     	\caption{R-MYULA: standard deviation (log-scaled) at different scales.}
    \end{subfigure}
    \begin{subfigure}[t]{\linewidth}
        \centering
    	\includegraphics[width=\textwidth]{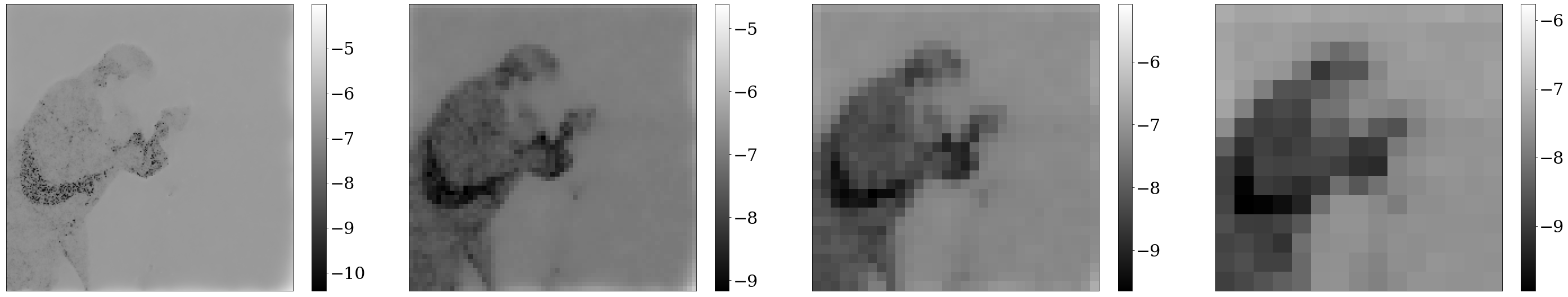}
    	\caption{R-SKROCK: standard deviation (log-scaled) at different scales.}
    \end{subfigure} 
    \begin{subfigure}[t]{\linewidth}
        \centering
    	\includegraphics[width=\textwidth]{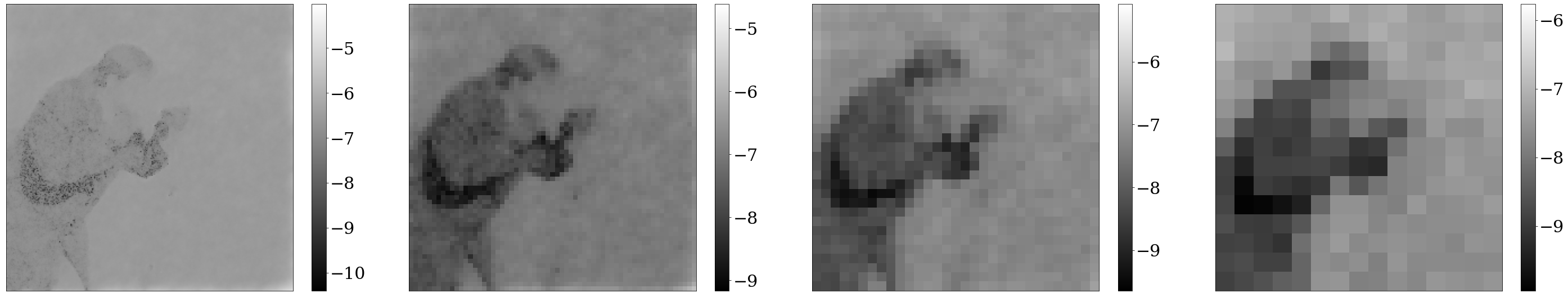}
    	\caption{R-MYUULA: standard deviation (log-scaled) at different scales.}
    \end{subfigure} 
    \caption{Marginal posterior standard deviation for the geometric inpainting problem at different scales (0,2,4,8 from left to right). The scale $i$ corresponds to a downsampling by a factor $2i$ of the original sample size. Most of the uncertainty is located at places where the likelihood is highly uninformative e.g. at the cameraman's background.}
    \label{fig:st_dev_geo}
\end{figure}

Figure \ref{fig:post_mean_geo} presents the MMSE estimates given by the proposed methods and the MAP estimate calculated by the adjusted PIDAL for geometric likelihoods. The MMSE estimates given by R-MYULA, R-SKROCK, R-MYUULA are visually similar with an apparent staircase effect while the MAP estimate returns an oversmoothed reconstruction. Since R-PMALA failed to properly converge in the available computing budget, it returned a noisy reconstruction. The PSNR values of each method are reported in the respective captions in Figure \ref{fig:post_mean_geo}. The comparison between MMSE and MAP estimation is more significant in this challenging experiment. Specifically, the PSNR values given by the MMSE estimates of R-MYULA, R-SKROCK, R-MYUULA are significantly higher than the PSNR value of the MAP estimator. 

Figure \ref{fig:NRMSE_geo} presents the evolution of the NRMSE estimation for the MMSE solutions as a function of the number of gradient evaluations (in log-scale). It can be observed that R-SKROCK has again the highest convergence speed, R-MYULA and R-MYUULA have similar convergence speed, and R-PMALA is by far the slowest algorithm.

The estimates of the pixel-wise standard deviations ($j=0$) are presented in the first column of Figure  \ref{fig:st_dev_geo}. It is observed that the estimates obtained by R-MYULA and R-MYUULA are slightly less accurate of the SKROCK estimate. High uncertainty is concentrated around the edges (e.g. cameraman figure ) and the cameraman's background, since the likelihood is highly uninformative for these pixel regions. For higher scales ($j\neq0$), the highest uncertainty is spotted at the background part of the cameraman's image as expected.

To conclude this experiment, in Figure \ref{fig:autocor_geo} we see that independence is reached fast for the components of low or median uncertainty, and is much slower for the few very uncertain pixels. Similarly to the previous experiments, R-SKROCK has better convergence properties compared with the other methods, see the decay of the ACF samples along the median and slowest components.

\begin{figure}[h!]
    \centering
    \begin{subfigure}[t]{0.32\linewidth}
        \centering
    	\includegraphics[width=\textwidth]{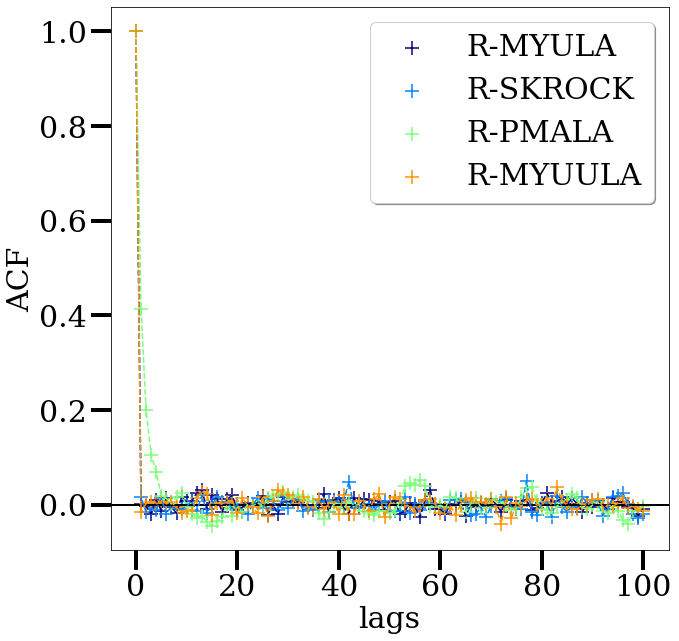}
     	\caption{Fastest component}
    \end{subfigure} 
    \begin{subfigure}[t]{0.32\linewidth}
        \centering
    	\includegraphics[width=\textwidth]{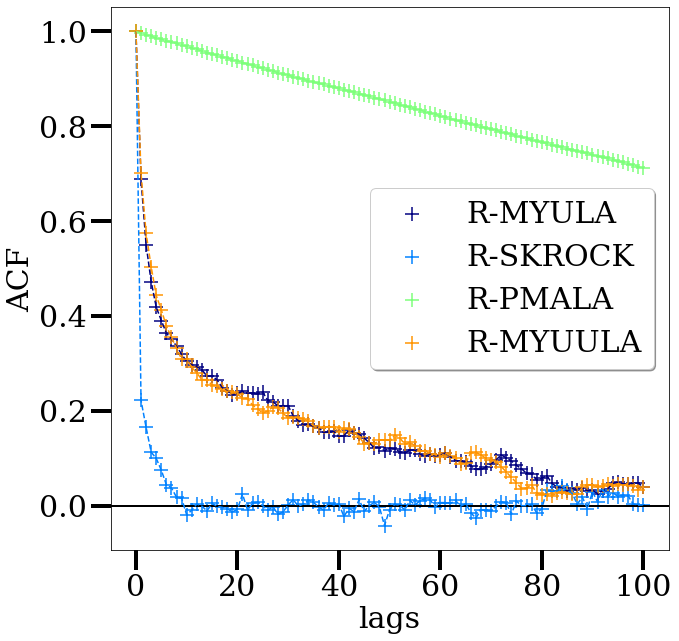}
    	\caption{Median component}
    \end{subfigure} 
    \begin{subfigure}[t]{0.32\linewidth}
        \centering
    	\includegraphics[width=\textwidth]{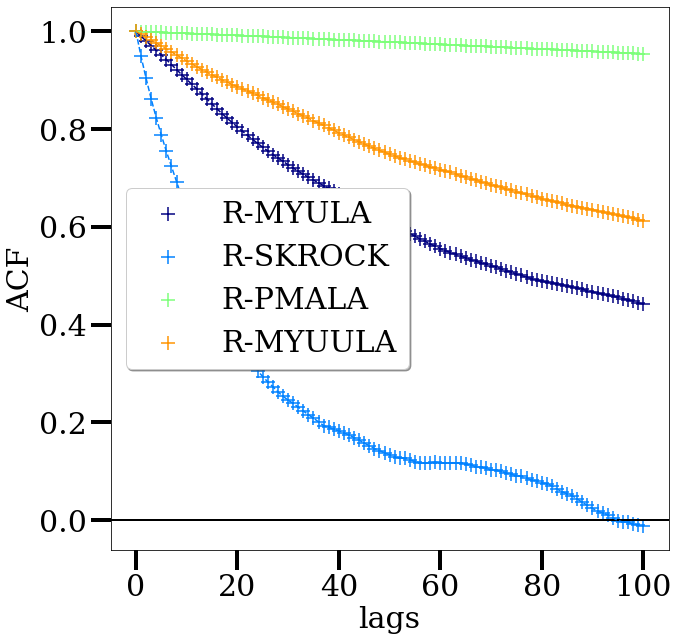}
    	\caption{Slowest component}
    \end{subfigure} 
    \caption{ACF for the fastest, median, and slowest components for the geometric inpainting problem. The ACF is shown for lags up to 100 for all images in the pixel domain.}
    \label{fig:autocor_geo}
\end{figure}

\section{Conclusion}

This work presented a new MCMC methodology based on a reflected and regularised Langevin SDE, specialised for performing Bayesian computation in low-photon imaging inverse problems involving a non-negativity constraint on the solution space. This reflected SDE is shown to be well-posed, have the desired invariant distribution, and to be exponentially ergodic under mild and easily verifiable conditions. This allowed deriving three unadjusted Langevin MCMC algorithms to perform Bayesian computation, as well as a Metropolised Langevin algorithm. The proposed approach was demonstrated with a range of experiments involving binomial, geometric, and low-intensity Poisson noise processes. In these experiments, we compared the computational accuracy and efficiency of the proposed MCMC algorithms, and illustrated their use for point estimation as well as uncertainty visualisation analyses. Moreover, we observed that in all the experiments with weakly informative data, the MMSE estimator computed by MCMC sampling outperformed MAP estimation (by convex optimization \cite{5492199}) in terms of accuracy. The methods also compared favourably against the state-of-the-art SPA proximal MCMC algorithm in the context of Poisson image deblurring \cite{8683031}. Future works will include a more detailed theoretical analysis of the proposed methodology and algorithms (e.g. non-asymptotic convergence rates \cite{durmus2016efficient}), as well as new variants suitable for data-driven priors such as plug-and-play priors encoded by image denoising neural networks \cite{laumont2021bayesian} and priors encoded by generative neural network models \cite{matts_paper}.

\appendix

\section{Proof of exponential ergodicity}\label{app:proofs_exp_ergodicity}

\begin{proof}[Proof of Lemma \ref{lem:TVconvergence}]
    Observe that $f_y^b(x), g^\lambda(x)$ both converge pointwise to $f_y(x),g(x)$ respectively as $(\lambda,b)\to (0,0)$. For $f_y^b$ this follows since $F$ is continuous and for $g^\lambda$ this follows from \cite[Equation (10)]{durmus2016efficient}. Moreover from \cite{durmus2016efficient} we have $g^\lambda \leq g$ and $f_y^b$ is lower bounded by assumption. Therefore by the dominated convergence theorem we have $\int e^{-f_y^b-g^\lambda} dx \to \int e^{-f_y-g} dx$ as $(\lambda,b)\to (0,0)$. Then by Scheff\'e's Lemma, see \cite[(5.10)]{williams1991probability}, we have 
    \begin{equation*}
        \lim_{\lambda\to 0,b\to 0} \int_{\R_+^d} \lvert \pi(x)-\pi^{\lambda,b}(x) \rvert dx =0.
    \end{equation*}
    That is, $\pi^{\lambda,b}$ converges to $\pi$ in total variation.
\end{proof}

\begin{proposition}\label{prop:smoothing}
    Let $\cP_t$ be the semigroup corresponding to the SDE \eqref{eq:RSDE}. Assume that $\nabla_x\Ulb$ is continuously differentiable and globally Lipschitz with constant $L\geq 0$. Then for all $\varphi\in C_b(\R_+^n)$ we have $\cP_t\varphi\in C_b^1(\R_+^n)$ for all $t>0$ and moreover there exists $C>0$ such that for all $t>0$ and $\varphi\in C_b(\R_+^n)$ we have
    \begin{equation*}
        \lVert \nabla_x\cP_t\varphi\rVert_\infty \leq \frac{Ce^{Ltn}}{\sqrt{t}}\lVert \varphi\rVert_\infty.
    \end{equation*}
\end{proposition}

\begin{proof}[Proof of Proposition \ref{prop:smoothing}]
    Fix $\varphi\in C_b(\R_+^n)$ then by \cite{DEUSCHEL} we have that $\cP_t \varphi\in C^1(\R_+^n)$ and we have the following representation by \cite[Theorem 2]{DEUSCHEL}
    \begin{equation}\label{eq:BEL}
        \frac{\partial}{\partial x_j}\cP_t\varphi(x) = \frac{1}{\sqrt{2}t}\sum_{i=1}^d\mathbb{E}\left[\textbf{E}_{0,j}\left[\varphi(X_t^x)\int_0^t\mathbbm{1}_{\xi_s=i}\mathbbm{1}_{\tau>s}\rho_{0,s} dW_s^i\right]\right].
    \end{equation}
    Here $\textbf{E}_{0,j}$ is the expectation of the process $\xi_s$ for with $\xi_0=j$, and $\xi_s$ is a continuous time Markov chain on $\{1,\ldots,d\}$ with generator
    \begin{equation*}
        \textbf{L}\varphi(i)=\sum_{k=1}^d \lvert c_t(i,k)\rvert[\varphi(k)-\varphi(i))], \quad  c_t(i,k)=\frac{\partial^2 \Ulb}{\partial_{x_i}\partial_{x_j}}(X_t^x).
    \end{equation*}
    Let $\eta_{\ell}$ be the sequence of jump moments of $\xi$, i.e. $\eta_{\ell+1}=\inf\{s>\eta_\ell:\xi_s\neq\xi_{\eta_\ell}\}$. Let $\tau$ be the stopping time $$\tau=\inf\{s>0:X_s^{x,\xi_s}=0\}.$$ With this notation we can define $\rho_{0,s}$ as
    \begin{equation*}
        \rho_{0,s}=\exp\left(\int_0^s\sum_{k\neq \xi_r}\lvert c_r(\xi_r,k)\rvert dr + \int_0^sc_r(\xi_r,\xi_r)dr\right)\prod_{0<\eta_k\leq s}\mathrm{sign}(c_{\eta_{k}}(\xi_{\eta_{k-1},\xi_{\eta_k}})).
    \end{equation*}
    Since $\nabla_x\Ulb(x)$ is globally Lipschitz we have that $c_t(i,k)$ is bounded in $(t,i,k)$, i.e. there exists $L>0$ such that $c_t(i,k)\leq L$ which gives the following bound on $\rho_{0,s}$:
    \begin{equation*}
        \lvert\rho_{0,s}\rvert\leq e^{sLn} .
    \end{equation*}
    Using this bound in \eqref{eq:BEL} we have
    \begin{equation*}
        \left\lvert\frac{\partial}{\partial x_j}\cP_t\varphi(x)\right\rvert \leq \frac{e^{tLn}}{t}\lVert \varphi\rVert_\infty\sum_{i=1}^d\mathbb{E}\left[\lvert W_s^i\rvert\right]=\sqrt{\frac{2}{\pi}}\frac{e^{tLn}n}{\sqrt{t}}\lVert f\rVert_\infty.
    \end{equation*}
\end{proof}

As a consequence of the above smoothing result we show that every compact set is small. 

\begin{definition}
A set $C\subseteq \R_+^n$ is called a small set if there exist $t>0$ a probability measure $\nu$ with $\nu(C)=1$ and positive constant $\eta>0$ such that
\begin{equation*}
    \cP_t\mathbbm{1}_E(x)\geq \eta \nu(E) \quad \text{ for all } x\in C, E\subseteq \R_+^n \text{ measureable}.
\end{equation*}
\end{definition}

\begin{proposition}\label{prop:small}
    Let $\cP_t$ be the semigroup corresponding to the SDE \eqref{eq:RSDE}. Assume that $\nabla_x\Ulb$ is continuously differentiable and globally Lipschitz. Then every compact set $K\subseteq \R_+^n$ is small.
\end{proposition}

\begin{proof}[Proof of Proposition \ref{prop:small}]
   First we show using a Girsanov transformation that the law of $X_t^x$ is equivalent to the law of a reflected Brownian motion and hence admits a positive density on $\R_+^n$ with respect to Lebesgue measure. Define the exponential martingale
\begin{equation*}
    Z_t=\exp\left(\frac{1}{\sqrt{2}}\int_0^t \nabla \Ulb(X_s)dW_s-\frac{1}{4}\int_0^t \nabla_x \Ulb(X_s)^2ds\right)
\end{equation*}
we observe that this is a true martingale by Novikov's condition which is satisfied since $\nabla_x \Ulb(X_s^x)$ is bounded. Then by Girsanov's theorem 
\begin{equation*}
\tilde{W}_t=W_t-\frac{1}{\sqrt{2}}\int_0^t\nabla \Ulb(X_s)ds
\end{equation*}
is a Brownian motion under $\mathbb{Q}$ where $\dfrac{d\mathbb{Q}}{d\mathbb{P}}=Z_t$. Therefore under $\mathbb{Q}$ we have 
\begin{equation*}
    X_t=x-\int_0^t\nabla \Ulb(X_s)ds +\sqrt{2}W_t +L_t = x+\sqrt{2}\tilde{W}_t+L_t
\end{equation*}
is a reflected Brownian motion. By \cite[Lemma 7.9]{Harrison} we have for any measurable $E\subseteq \R_+^N$ and $x\in \R_+^N$ that $\mathbb{Q}(X_t^x\in E)=0$ if and only if $E$ has zero Lebesgue measure. Hence $\mathbb{P}(X_t^x\in E)=0$ if and only if $E$ has zero Lebesgue measure since $\mathbb{P}$ is equivalent to $\mathbb{Q}$. In particular we have that the law of $X_t^x$ admits a density $p_t(\cdot;x)$ with respect to Lebesgue measure on $\R_+^n$ and the density is strictly positive.

Fix a compact set $K\subseteq \R_+^n$ and $t>0$. Define the measure $\nu$ by
\begin{equation*}
    \nu(E)=\frac{\inf_{x\in K}\mathbb{P}(X_t^x\in E\cap K)}{\inf_{x\in K}\mathbb{P}(X_t^x\in K)}, \text{ for any measurable }  E\subseteq \R_+^n,
\end{equation*}
it remains to show that $\nu$ is a probability measure, i.e. that $\inf_{x\in K}\mathbb{P}(X_t^x\in K)>0$. Since the law of $X_t^x$ admits a strictly positive density we know that $\mathbb{P}(X_t^x\in K)>0$ for each $x\in K$. If we have that $x\mapsto \mathbb{P}(X_t^x\in K)$ is continuous then since $K$ is compact, $\mathbb{P}(X_t^x\in K)$ attains its minimum value and that value must be positive, i.e. $\inf_{x\in K}\mathbb{P}(X_t^x\in K)>0$. It remains to show $x\mapsto \mathbb{P}(X_t^x\in K)$ is continuous. Let $\varphi(x)=\mathbbm{1}_K(x)$ and approximate this function by the sequence $\{\varphi_k\}_{k\geq 0}$ where 
\begin{equation*}
    \varphi_k(x)=\max\left(1-\frac{1}{k}d(x,K),0\right)
\end{equation*}
and $d(x,K)$ denotes the distance between $x$ and the set $K$. Note that $\varphi_k\in C_b(\R_+^n)$, $\lVert \varphi_k\rVert_\infty\leq 1=\lVert \varphi\rVert_\infty$ and $\varphi_k(x)\to \varphi(x)$ as $k\to\infty$ for each $x\in \R_+^n$. By the dominated convergence theorem $\cP_t\varphi_k(x)\to\cP_t\varphi(x)$ for each $x\in \R_+^n,t>0$. On the other hand by Proposition \ref{prop:smoothing} the derivative of $\cP_t\varphi_k$ is bounded uniformly in $k$ and $x$ and in particular, for fixed $t>0$ the sequence $\{\cP_t\varphi_k\}$ is uniformly bounded and equicontinuous therefore by Arzel\'a–Ascoli Theorem there exists some $\psi\in C_b(\R_+^n)$ such that $\cP_t\varphi_k$ converges locally uniformly to $\psi$. Since $\cP_t\varphi_k$ converges pointwise to $\cP_t\varphi$ we have that $\psi=\cP_t\varphi$ and hence $\cP_t\varphi\in C_b(\R_+^n)$.
\end{proof}

Before we prove Theorem \ref{thm:expergodicityconvex} we need the following elementary lemma from \cite[Lemma 2.2]{Bakrysimpleproof}. Note that although this Lemma is proved for functions on the whole space $\R^n$ the proof still holds when restricted to the space $\R^n_+$.
\begin{lemma}\label{lem:convexbound}
If $\Ulb$ is differentiable, convex and $\int_{\R^n_+} e^{-\Ulb(x)}dx<\infty$ then there exist $\alpha>0$ and $R>0$ such that for $\lvert x\rvert \geq R$, \eqref{eq:convexlowerbound} holds.
\end{lemma}

\begin{proof}[Proof of Theorem \ref{thm:expergodicityconvex}]
    By Proposition \ref{prop:small} we have that every compact set is small and hence is petite (by \cite[Proposition 5.5.3]{meyn2012markov}) therefore to verify the conditions of Theorem \ref{thm:MeynTweedie} it remains to find a function $V(x)$ which satisfies Hypothesis \ref{hyp:Lyapunov}. 
    Set
    \begin{equation*}
        V(x)=\exp\left(\gamma (1+\lVert x\rVert^2)^{\frac{1}{2}}\right).
    \end{equation*}
    Note that
    \begin{equation*}
        \partial_{x_i}V(x) = \frac{\gamma_ix_i}{(1+\lVert x\rVert^2)^{\frac{1}{2}}}V(x)
    \end{equation*}
    In particular, we see that $\partial_{x_i}V(x)=0$ if $x_i=0$ and hence $V$ is in the domain of the extended generator $\cL_\ell$ for each $\ell$. It remains to show that $V(x)$ satisfies \eqref{eq:Lyapunovcondition}. For any $x\in O_\ell$ we have
    \begin{align}
    \frac{\cL_\ell V(x)}{V(x)} = -\frac{\gamma}{(1+\lVert x\rVert^2)^\frac{1}{2}}\langle\nabla_xU(x),x\rangle+\frac{\gamma n}{(1+\lVert x\rVert^2)^\frac{1}{2}}+\frac{\gamma^2\lVert x\rVert^2}{1+\lVert x\rVert^2}.
    \end{align}
    For $\lvert x\rvert\geq R$, using Lemma \ref{lem:convexbound} we have
    \begin{align*}
    \frac{\cL_\ell V(x)}{V(x)} &\leq -\frac{\gamma}{(1+\lVert x\rVert^2)^\frac{1}{2}}\alpha\lVert x\rVert+\frac{\gamma n}{(1+\lVert x\rVert^2)^\frac{1}{2}}+\frac{\gamma^2\lVert x\rVert^2}{1+\lVert x\rVert^2}.
    \end{align*}
    Since $z/\sqrt{1+z^2}$ is increasing for $z>0$ we can use that $\lVert x\rVert\geq R$ to obtain the bound
    \begin{align}
    \frac{\cL_mV(x)}{V(x)} \leq \gamma\left(\gamma-\alpha\frac{R}{(1+R^2)^\frac{1}{2}}\right)+\frac{\gamma n}{(1+R^2)^\frac{1}{2}}.
    \end{align}
    Set $\gamma=\alpha/4$ and let $R$ be sufficiently large that $\frac{\gamma n}{(1+R^2)^\frac{1}{2}}\leq \alpha/4$, $R/(1+R^2)^{\frac{1}{2}}\geq 3/4$ and that $\frac{\gamma n}{(1+R^2)^\frac{1}{2}} \leq \alpha^2/16$ then we have
    \begin{equation*}
        \frac{\cL_\ell V(x)}{V(x)} \leq -\frac{1}{16}\alpha^2.
    \end{equation*}
    Therefore Hypothesis \ref{hyp:Lyapunov} is satisfied and by Theorem \ref{thm:MeynTweedie} we have that $X_t^x$ is exponentially ergodic.

    In order to show that $\pi$ is given by \eqref{eq:pilb} it is sufficient to show that for some core $\mathcal{C}$ of the domain of $\cL$ we have
    \begin{equation*}
        \int_{\R_+^n} \cL \varphi(x) e^{-\Ulb(x)} dx=0 \quad \text{ for all } \varphi\in \mathcal{C}.
    \end{equation*}
    Here we are viewing $\cL$ as the infinitesimal generator of the semigroup $\cP_t$ in the space $C_0(\R_+^n)$ (i.e. the set of continuous functions vanishing at infinity endowed with the supremum norm).
    Let 
    \begin{equation*}
        \mathcal{C}=\{\varphi\in C^\infty((0,\infty)^n)\cap C^1_b(\R_+^n): \cL \varphi\in C_0(\R_+^n), \, \partial_{x_i}\varphi(x)=0 \text{ if } x_i=0\}.
    \end{equation*}
    That is $\mathcal{C}$ is the set of functions which are infinitely differentiable in the interior of $\R_+^n$, are $C^1$ up to the boundary, is bounded with bounded first order derivative, the generator applied to $\varphi$ belongs to $C_0(\R_+^n)$ and satisfies the boundary conditions of the PDE \eqref{eq:NPDE}.
    Since $\cP_t$ preserves the set $C$, and $C$ is contained within the domain of $\cL$ and is dense in $C_0(\R_+^n)$ we have that $C$ is a core of $\cL$ by \cite[Proposition 3.3]{Ethier}. Fix $f\in C$ then by integration by parts
    \begin{align*}
        \int_{\R_+^n} \cL \varphi(x)e^{-\Ulb(x)} dx&= \sum_{i=1}^n \int_{\R_+^n} \left(\partial_{x_i}\Ulb(x)\partial_{x_i} \varphi(x) +\partial_{x_i}^2\varphi(x)\right) e^{-\Ulb(x)}dx\\
        &= \sum_{i=1}^n \int_{\R_+^n} [-\partial_{x_i}\Ulb(x)\partial_{x_i} \varphi(x) +\partial_{x_i}\Ulb(x)\partial_{x_i} \varphi(x)] e^{-\Ulb(x)}dx\\
        &+\left[\partial_{x_i}\varphi(x)e^{-\Ulb(x)}\right]\big\rvert_{x_i=0}\\
        &=0.
    \end{align*}
\end{proof}

%

\end{document}